\documentclass[sageh,Crown]{sagej}

\usepackage{moreverb,url}
\usepackage{xparse}
\newcommand{\bnm}{\begin{newmath}}
\newcommand{\enm}{\end{newmath}}

\newcommand{\bea}{\begin{eqnarray*}}%
\newcommand{\eea}{\end{eqnarray*}}%

\newcommand{\bne}{\begin{newequation}}
\newcommand{\ene}{\end{newequation}}

\newcommand{\bal}{\begin{newalign}}
\newcommand{\eal}{\end{newalign}}

\newenvironment{newalign}{\begin{align}%
\setlength{\abovedisplayskip}{4pt}%
\setlength{\belowdisplayskip}{4pt}%
\setlength{\abovedisplayshortskip}{6pt}%
\setlength{\belowdisplayshortskip}{6pt} }{\end{align}}

\newenvironment{newmath}{\begin{displaymath}%
\setlength{\abovedisplayskip}{4pt}%
\setlength{\belowdisplayskip}{4pt}%
\setlength{\abovedisplayshortskip}{6pt}%
\setlength{\belowdisplayshortskip}{6pt} }{\end{displaymath}}

\newenvironment{newequation}{\begin{equation}%
\setlength{\abovedisplayskip}{4pt}%
\setlength{\belowdisplayskip}{4pt}%
\setlength{\abovedisplayshortskip}{6pt}%
\setlength{\belowdisplayshortskip}{6pt} }{\end{equation}}

\newcounter{ctr}

\newcounter{mytable}
\def\mytable{\begin{centering}\refstepcounter{mytable}}
\def\endmytable{\end{centering}}

\newcounter{myfig}
\def\myfig{\begin{centering}\refstepcounter{myfig}}
\def\endmyfig{\end{centering}}

\newlength{\saveparindent}
\newlength{\saveparskip}
\newcommand{\E}{{\rm I\kern-.3em E}}

\def\E{{\bf E}}

\def\0{{\bf 0}}
\def\1{{\bf 1}}

\def\ll{\llbracket}
\def\rr{\rrbracket}

\renewcommand{\eqref}[1]{\mbox{Equation~(\ref{#1})}}

\def \part {part}

\renewcommand{\paragraph}[1]{\vspace*{6pt}\noindent\textbf{#1}\;}

\def \blackslug{\hbox{\hskip 1pt \vrule width 4pt height 8pt
    depth 1.5pt \hskip 1pt}}
\def \qed{\quad\blackslug\lower 8.5pt\null\par}

\newcommand{\rev}[1]{{#1}}

\newcommand\ignore[1]{}

\renewcommand{\update}{\textsf{Insert}}

\newcounter{rcnote}[section]

\newcounter{mrnote}[section]

\newcounter{fknote}[section]

\newcounter{anote}[section]

\DeclareMathSymbol{\mlq}{\mathord}{operators}{``}
\DeclareMathSymbol{\mrq}{\mathord}{operators}{`'}

\newcommand{\rhf}[2]{R_{f, \gamma}}

\DeclareDocumentCommand{\edist}{o o}{
  \ensuremath{
    \IfNoValueTF{#1}{{d}}{{\sf d}(#1,#2)}
  }
}

\newcommand{\olrk}[1]{\ifx\nursymbol#1\else\!\!\mskip4.5mu plus 0.5mu\left(\mskip0.5mu plus0.5mu #1\mskip1.5mu plus0.5mu \right)\fi}

\providecommand{\add}{\funcfont{Add}}

\NewDocumentCommand{\indseq}{ O{1} O{r} }{{#1}\ldots {#2}}

\usepackage{enumitem}
\usepackage{booktabs}
\usepackage{tabularx}
\usepackage{xcolor}
\usepackage{caption}
\usepackage{colortbl}
\usepackage{array}
\usepackage{adjustbox}
\usepackage{tikz}
\usepackage{gensymb}
\usepackage{amsmath}
\usepackage{amssymb}
\usepackage{amsthm}
\usepackage{dsfont}
\usepackage{mathrsfs}
\usepackage{stmaryrd}
\usepackage{dutchcal}
\usepackage{bm}
\usepackage{threeparttable}
\usepackage{mathtools}
\usepackage{hyperref}
\usepackage{algorithm}
\usepackage{algorithmicx}
\usepackage{algpseudocode}
\usepackage{rotating}
\usepackage{pifont}
\usepackage{pdflscape}

\theoremstyle{plain}
\newtheorem{theorem}{Theorem}[section]
\newtheorem{lemma}[theorem]{Lemma}
\newtheorem{corollary}[theorem]{Corollary}
\newtheorem{example}{Example}
\newtheorem{definition}{Definition}

\AtBeginDocument{%
  \providecommand\BibTeX{{%
    Bib\TeX}}}
\newcommand\BibTeX{{\rmfamily B\kern-.05em \textsc{i\kern-.025em b}\kern-.08em
T\kern-.1667em\lower.7ex\hbox{E}\kern-.125emX}}

\begin{document}
\setcounter{secnumdepth}{4}
\def\thetitle{SLIDE: \underline{S}huffle Shamir Secret Shares Uniformly with \underline{Li}near Online Communication and Guaranteed Output \underline{De}livery}
\title{\thetitle}

\author{Jiacheng Gao\affilnum{1}\textsuperscript{$\dagger$}, Moyang Xie\affilnum{1}\textsuperscript{$\dagger$},
Yuan Zhang\affilnum{1} and Sheng Zhong\affilnum{1}}

\affiliation{\affilnum{1}State Key
Laboratory for Novel Software Technology, Nanjing University, China
\textsuperscript{$\dagger$}Jiacheng Gao (jcgao@smail.nju.edu.cn) and Moyang Xie (xie\_moyang@foxmail.com) contributed equally to this work.
}

\corrauth{Yuan Zhang and Sheng Zhong, Gulou Campus, Nanjing University, Hankou Road 22, Gulou District, Nanjing, Jiangsu Province, China, Postcode: 210093}

\email{zhangyuan@nju.edu.cn, zhongsheng@nju.edu.cn}

\begin{abstract}
We revisit shuffle protocols for Shamir secret sharing.
Existing constructions either produce non-uniform shuffles or incur high communication and round complexity,
sometimes exponential in the number of parties.
We propose two new shuffle protocols that achieve uniform shuffling with communication complexity
$O\!(\frac{k + l}{\log k} \, n^2 m \log m)$
for an $m \times l$ matrix shared among $n$ parties,
where $k \leq m$ is a tunable parameter.
The first protocol is concretely efficient, while the second achieves the best-known $O(nml)$ online communication and $O(n)$ rounds.
Experiments show significant improvements in online efficiency and total cost over prior work.
Our key technical ingredient is a novel permutation sharing technique that represents permutations via smaller permutation matrices,
making their application significantly more efficient.
The first protocol applies independent secret permutations sequentially,
while the second builds on shuffle correlation to achieve optimal online complexity.
We further extend shuffle correlation to support GOD with linear online communication,
yielding \textsf{SLIDE}, the first protocol achieving both $O(nml)$ online communication and guaranteed output delivery.
Our constructions rely only on basic Shamir secret sharing over any field $\mathbb{F}$ with $\lvert\mathbb{F}\rvert > n$.
As shuffling is a fundamental primitive for MPC tasks such as sorting and oblivious data structures,
our results enable more efficient and scalable secure computation in practice.
\end{abstract}

\keywords{Secure Multiparty Computation, Secure Shuffle, Shamir Secret Sharing}

\maketitle

\section{Introduction}\label{sec:introduction}
Secure multiparty computation (MPC) has found wide applicability in
real-world scenarios. In an MPC setting, multiple parties jointly evaluate
a function over their private inputs. The goal is to preserve the
confidentiality of each party's data while efficiently obtaining the
correct output. This strong security guarantee makes MPC an attractive
solution for many privacy-sensitive collaborative tasks.

Most MPC protocols are designed modularly, with existing primitives
serving as subroutines for building new constructions. This modularity
enables succinct protocol descriptions and clean security arguments via
the composition theorem \cite{canetti2000composition}. Among these
primitives, the MPC shuffle protocol is particularly fundamental.
An MPC shuffle protocol takes an array of secret-shared values as input,
randomly permutes it, and shares the result back to the parties such that
no party knows which output entry comes from which input entry. Such an
operation has numerous real-world applications, including anonymous
communication systems \cite{alexopoulos2017mcmix, lu2019shamirshuffle, eskandarian2021clarion},
electronic voting \cite{kusters2020evoting, jingzhong2020evoting, pedin2023evoting},
and distributed databases \cite{zhang2023database, wang2022database}.
The shuffle operation is also a versatile building block for more complex
MPC primitives, including MPC sorting
\cite{hamada2013shufflethensort, hamada2014radixsort, laud2016shufflethensort, bogdanov2014radixsort}
and MPC oblivious RAM (ORAM)
\cite{keller2014ORAM, peceny2025permutation, falk2023doram}.

Despite the broad applicability of MPC shuffles, recent research has
focused mainly on improving the efficiency of their additive-secret-sharing
variants. A major breakthrough was achieved by \cite{chase2020secret},
who proposed a novel semi-honest two-party shuffle protocol with almost
linear computation and communication. Follow-up studies
\cite{eskandarian2021clarion, laud2021linearshuffle, song2023secret, gao2024linearshuffle}
extended this line of work to active security and the multi-party setting,
achieving linear online communication and computation. As a result,
shuffle protocols based on additive secret sharing are already highly
efficient.

In contrast, shuffle protocols for Shamir secret sharing \cite{shamirss}
remain far less developed. The most efficient known approaches rely on
switching networks \cite{lu2019shamirshuffle, mardi2021shuffle} or
sorting networks \cite{movahedi2015shuffle1, movahedi2015shuffle2}, with
$O(nml\log m)$ communication for $n$ parties shuffling $m$ vectors of
length $l$. However, these constructions exhibit significant drawbacks:
switching-network-based shuffles fail to generate uniform shuffles
\cite{lu2019shamirshuffle, mardi2021shuffle}, while
sorting-network-based shuffles achieve uniformity only with some
probability \cite{movahedi2015shuffle1, movahedi2015shuffle2}. These
limitations complicate the security analysis of higher-level MPC protocols
built upon them.

Currently, uniform shuffling in Shamir secret sharing is known only through
the permute-in-turn paradigm
\cite{laur2011shuffle, keller2014ORAM}, which incurs substantially higher
complexities. The construction of \cite{laur2011shuffle} requires
$O(2^nn^{1.5})$ rounds, rendering it practical only for very small numbers
of parties. The construction of \cite{keller2014ORAM} requires
$O(n^2ml\log m)$ communication and $O(n\log m)$ rounds, which is still less
efficient than existing non-uniform shuffle protocols. These approaches
are further discussed in Section~\ref{sec:Related Works}, with a comparison
of existing constructions provided in Table~\ref{table:shuffle protocols}.

\paragraph{Contributions.}
This paper develops efficient and uniform shuffle protocols for Shamir
secret sharing, with unconditional security against malicious parties
under an honest majority. These protocols are further enhanced to achieve
guaranteed output delivery (GOD).

First, we adopt the permute-in-turn approach to guarantee the uniformity
of the shuffle. At the core of our contributions is a novel permutation
sharing technique that enables the parties to share a permutation and
apply it to secret-shared data. Our technique requires
$O(\tfrac{nmk\log m}{\log k})$ communication to share an $m$-permutation
and $O(\tfrac{nml\log m}{\log k})$ communication to apply the shared
permutation to $m$ vectors of length $l$; here, $1\leq k\leq m$ is a
tunable parameter balancing communication and computation. By letting
each party share one secret, uniformly random permutation in the offline
phase and sequentially applying these permutations in the online phase,
we obtain our first shuffle protocol.

To further improve online efficiency, we leverage the shuffle correlation
technique of \cite{gao2024linearshuffle}. The original shuffle correlation
of Gao et al. achieves only security with abort: honest parties abort
without obtaining the output if the adversary deviates from the protocol.
We therefore investigate how to instantiate their construction using our
permutation sharing protocol. We show that, with the Shamir secret sharing
scheme of \cite{goyal2020guaranteed}, the dispute control technique
\cite{beerliova2006dispute}, and the authentication tag technique of
\cite{Ben-Sasson2012miniMPC}, we can securely generate and apply a shuffle
correlation with guaranteed output delivery (GOD). This yields
\textsf{SLIDE}, the first uniform shuffle protocol that achieves both GOD
and the best-known online communication complexity of $O(nml)$, along with
$O(n)$ rounds.

Our contributions are summarized as follows.
    \begin{enumerate}
\item We introduce a novel primitive for sharing a permutation and applying
    it to secret-shared vectors. Using this primitive, we construct a
    shuffle protocol with
    $O(\tfrac{(k+l)n^2m\log m}{\log k})$ total communication and
    $O(\tfrac{n\log m}{\log k}+\log\log k)$ rounds, where $k$ is a
    tunable parameter. This protocol outperforms all existing uniform
    shuffle protocols for Shamir secret sharing. Moreover, when
    $n=o(\log m)$, choosing $k=\omega(2^n)$ yields
    $o(nml\log m)$ online communication and $o(\log m)$ online rounds.
    Thus, in this parameter regime, the resulting protocol improves upon
    even switching- and sorting-network-based approaches, at the expense
    of higher offline communication (Section~\ref{sec:permutation}).
\item We enhance our shuffle protocol with the shuffle correlation
    technique of \cite{gao2024linearshuffle}. This yields a second
    shuffle protocol with the best-known $O(nml)$ online communication
    complexity and $O(n)$ rounds. Its offline communication complexity,
    $O(\tfrac{(k+l)n^2 m \log m}{\log k})$, also improves upon the previous
    best-known bound of $O(n^2 ml \log m)$ under the uniformity requirement.
    We further implement both protocols and conduct extensive experiments,
    showing that our constructions, particularly \textsf{Shuffle2}, achieve
    significant improvements in online efficiency and total cost over prior
    work (Section~\ref{sec:shuffle}).
\item We enhance both our protocols and the shuffle correlation of
    \cite{gao2024linearshuffle} to achieve GOD, ensuring that honest
    parties obtain the correct output regardless of adversarial behavior.
    This yields \textsf{SLIDE}, the first shuffle protocol that
    simultaneously achieves uniformity, GOD, and the best-known
    $O(nml)$ online communication. We note that whether shuffle correlation
    could be extended to GOD was not previously known, and our work
    establishes this as a new result (Section~\ref{sec:GOD}).
\item We analyze the communication, computation, and round complexities
    of our protocols to quantify the effect of $k$ and guide the choice
    of an appropriate value.
    \end{enumerate}

\paragraph{Organization.}

The rest of the paper is organized as follows.
In Section~\ref{sec:Related Works}, we review prior work on MPC shuffle
protocols, with a focus on constructions for Shamir secret sharing.
Section~\ref{sec:preliminary} introduces the basic notation and security
models used throughout the paper. In Section~\ref{sec:permutation}, we
present our permutation sharing technique and show how to build a shuffle
protocol from it. Section~\ref{sec:shuffle} demonstrates how to leverage
this technique to construct a shuffle protocol with linear online
complexity, matching the overall complexity of the protocol in
Section~\ref{sec:permutation} up to constant factors.
Section~\ref{sec:GOD} shows how to enhance our protocols to support
guaranteed output delivery (GOD), including how to adapt the shuffle
correlation of \cite{gao2024linearshuffle}.

For brevity, the security proofs are provided in
Appendix~\ref{sec:security proofs}, the proofs of our complexity analyses
are given in Appendix~\ref{sec:complexity}, and several auxiliary
discussions are deferred to Appendix~\ref{sec:discussion}.

\section{Related Works}\label{sec:Related Works}
\paragraph{Secure MPC Shuffle.}
The concept of a multi-party shuffle originates from
\cite{chaum1981untraceable} as a byproduct of the decryption shuffle.
In this scenario, clients iteratively encrypt messages using the servers'
public keys and send them to the first server. Each server then decrypts,
randomly permutes, and forwards the message batch. After the final server
processes the batch, the messages are securely shuffled, ensuring that no
server can link messages to their senders. While the original protocol
assumes honest servers, subsequent works achieve active security
\cite{adida2007shuffle, groth2010verifiable, hohenberger2007securely, abe1999mix}
through zero-knowledge proofs that verify the correctness of decryption
and permutation. These approaches incur $O(m^2)$ computation per server
and are generally heavy due to the overhead of public-key operations and
proof generation.

Despite its importance, the decryption shuffle technique is not well
suited for multi-party computation (MPC) tasks. After decryption and
permutation, its output consists of permuted plaintexts. In contrast, MPC
scenarios require the output to remain secret-shared so that subsequent
secure computations can be performed on the data. For instance, in
\cite{hamada2014radixsort} and \cite{bogdanov2014radixsort}, the authors
construct MPC sorting protocols based on shuffling, where it is essential
that the shuffle outputs remain secret-shared to enable further secure
computations.

A recent breakthrough by \cite{chase2020secret} significantly improved
the efficiency of two-party MPC shuffles, thereby enhancing the efficiency
of multi-party MPC and its shuffle-based applications. Their key idea is
a permutation decomposition technique that breaks a large $m$-permutation
into $O\!\big(\frac{m \log m}{k \log k}\big)$ small $k$-permutations.
Each $k$-permutation is implemented using obliviously punctured matrices
(OPMs) with $O(kl + k\log k)$ communication and $O(k^2 l)$ local
computation, yielding a two-party shuffle protocol with
$O\!\big((1 + \frac{l}{\log k}) m \log m\big)$ communication. At its core,
the protocol relies only on oblivious transfer extension (OTE) and a
pseudorandom generator (PRG), both of which are computationally
lightweight.

\cite{han2025permutation} and \cite{peceny2025permutation} identify the
randomness underlying \cite{chase2020secret} as permutation correlation.
\cite{han2025permutation} improves its concrete generation efficiency,
while \cite{peceny2025permutation} proposes a new method with
$O(m \log l)$ communication for generating such randomness under the
Learning Parity with Noise (LPN) assumption. Initially designed for the
two-party semi-honest setting, this line of work has since been extended
to the multi-party setting \cite{laud2021linearshuffle}, to linear online
communication \cite{eskandarian2021clarion, chen2025shuffle}, and to
malicious security \cite{song2023secret}. A recent advance by
\cite{gao2024linearshuffle} achieves both malicious security and linear
online communication in the multi-party setting.

However, all the above advances in MPC shuffle protocols are designed for
additive secret sharing and cannot be directly applied to Shamir secret
sharing. The gap stems from fundamental differences between the two
schemes. In Shamir secret sharing, shares must be points on a low-degree
polynomial, which excludes additive-sharing techniques such as PRG-based
rerandomization, used in \cite{chase2020secret} to balance computation
and communication. Consequently, shuffle protocols over Shamir secret
sharing face significantly stricter design constraints and efficiency
limitations.

\paragraph{Secure MPC Shuffle for Shamir Secret Sharing.}
To date, four main approaches are known for constructing MPC shuffle
protocols in Shamir secret sharing.

The first approach uses switching networks, introduced by
\cite{lu2019shamirshuffle} and optimized by \cite{mardi2021shuffle}.
In this approach, a shuffle is implemented via layers of switch gates:
each gate randomly swaps two inputs based on a random bit, and each gate
can be realized using MPC multiplication. \cite{mardi2021shuffle} achieve
a best-known communication complexity of $O(nml \log m)$. However, a
critical limitation is that the resulting permutation is not uniform, as
some permutations occur with significantly higher probability than others
(see Table~4 of \cite{mardi2021shuffle}), weakening the overall security
guarantee.

The second approach is based on sorting networks, initially proposed by
\cite{movahedi2015shuffle1,movahedi2015shuffle2}. The idea is to tag each
data item with a random value; sorting by these tags yields a uniform
shuffle. This reduces the problem to implementing an MPC sorting protocol.
A key drawback, however, is that MPC sorting protocols rely on heavy MPC
primitives. In addition, the design of sorting networks, or equivalently
oblivious sorting algorithms, is difficult. Existing networks are either
too deep (e.g., Batcher's odd--even merge sort and bitonic sort
\cite{batcher1968OddEvenMergeSort}, both with depth $O(\log^2 m)$), have
impractically large constants (e.g., the AKS network
\cite{ajtai1983AKSNetwork}), or sort correctly only with some probability
(e.g., the LP network \cite{leighton1990LPNetwork}). This also explains
why many MPC sorting protocols adopt the shuffle-then-sort paradigm
\cite{laud2016shufflethensort,hamada2013shufflethensort}, eliminating the
need for an oblivious sorting algorithm. Using the LP network,
\cite{movahedi2015shuffle1,movahedi2015shuffle2} achieve
$O(nml\log m)$ communication and a permutation that is uniform with
probability $1 - O(m^{-3})$. Although this leaves a small, albeit
non-negligible, probability of a non-uniform shuffle, this guarantee is
already sufficient for many applications.

The third approach is based on the square network, which was originally
proposed as a plaintext algorithm by~\cite{haastad2006square}.
Later, \cite{lu2023rpm} used it to realize a secure shuffle under
Shamir secret sharing. At a high level, the square network composes a
large permutation from a sequence of smaller permutations, resulting in a
distribution that is close to, but not exactly, uniform. Concretely,
\cite{lu2023rpm} adopts an approach that approximates the sampling of a
uniform permutation over $m$ elements by sampling $q n \sqrt{m}$
permutations over $\sqrt{m}$ elements, where $q$ denotes the number of
layers in the square network and is set to the constant $15$. After $q$
layers, the resulting permutation distribution $\Pi_q$ approaches the
uniform distribution $U_m$ in statistical distance. In particular, it
holds that
$
\Delta(\Pi_q, U_m) \leq
O\!(
m^{1 - \lfloor q/3 \rfloor \cdot \frac{1}{4}}
(\log m)^{\lfloor q/3 \rfloor}
\!).
$
To shuffle $m$ messages of length $l$, the parties perform
$q n l \sqrt{m}$ matrix multiplications, each involving two square
matrices of dimension $\sqrt{m}$. Since each such multiplication incurs
$O(n \cdot m^{1.5})$ computation in MPC, the total computation cost is
$O(q n^2 m^2 l)$.

The fourth approach is permute-in-turn, adopted by
\cite{laur2011shuffle} and \cite{keller2014ORAM}. In
\cite{laur2011shuffle}, the strategy is to have multiple subsets of
parties agree on independent uniform permutations, apply these
permutations to the array, and then reshare the result with all parties.
Each subset must be large enough to allow secret reconstruction, and
sufficiently many subsets are required so that an adversary cannot corrupt
all of them. Their analysis shows that approximately
$O(2^n / \sqrt{n})$ subsets are needed, resulting in a communication
complexity of $O(2^n n^{1.5} m l \log m)$.
\cite{keller2014ORAM} improve this paradigm by incorporating a switching
network. In their construction, each party sequentially applies a
permutation by sharing $O(m \log m)$ control bits that determine the
network's switching operations. As all parties permute sequentially, the
resulting shuffle is uniform. This reduces the communication cost to
$O(n^2 ml\log m)$ within $O(n\log m)$ rounds.
The primary advantage of the permute-in-turn paradigm is its guarantee of
uniformity. However, both variants incur higher communication and round
complexities than switching- or sorting-network-based approaches. The
construction of \cite{laur2011shuffle} requires a round complexity
exponential in $n$, while the construction of \cite{keller2014ORAM} adds
a factor of $n$ in communication complexity.

Beyond the aforementioned advantages and drawbacks, a crucial limitation
of existing works is their neglect of online complexity. Many MPC
protocols adopt a two-phase design consisting of a data-independent
offline phase and a data-dependent online phase. This phase separation
enables most heavy operations to be precomputed offline, leaving the
online phase lightweight so that the parties can respond swiftly to
real-world data as it arrives. In Appendix~\ref{subsec:Online Complexity},
we discuss why online complexity is often
more crucial than offline complexity.

However, nearly all existing shuffle protocols have only a trivial
offline phase that prepares basic randomness, such as shared random values
or bits. For instance, switching-network-based protocols require all
input data to be available before they can be fed into the network,
meaning that only random bits can be precomputed offline. Similarly,
sorting-based approaches cannot proceed without actual inputs, as there
is no data to sort. Even though the construction of \cite{keller2014ORAM}
allows parties to share and verify control bits offline, its online phase
still incurs $O(n^2 ml\log m)$ communication, matching its offline cost.
As a result, the online workload remains heavy and becomes a bottleneck
in large-scale MPC applications.

\begin{landscape}
\vspace*{\fill}
\begin{table}[h]
    \centering
\caption{Summary of existing MPC shuffle protocols based on Shamir secret sharing.
Here, $n$ is the number of parties, $m$ is the number of vectors to be
shuffled, $l$ is the vector length in field elements, $k$ is a tuning
parameter balancing communication, computation, and round complexity, and
$q$ is the number of layers in the square network, set to the constant
$15$ in~\cite{lu2023rpm}. $\mathcal{B}$ and $\mathcal{C}$ denote the
costs of generating one shared random bit and one MPC comparison,
respectively, as determined by the underlying implementation. ``Unif.''
stands for ``uniform''. ``Onl./Offl.'' denotes ``online/offline'', and
``Comm.'' and ``Rds.'' denote communication and round complexity,
respectively. The protocols of
\cite{mardi2021shuffle, movahedi2015shuffle1, movahedi2015shuffle2, keller2014ORAM}
are generic and not limited to Shamir secret sharing, although their
asymptotic costs may differ under other MPC frameworks.}
        \label{table:shuffle protocols}

        \setlength{\tabcolsep}{2.4pt}
        \footnotesize
        \begin{tabularx}{\linewidth}{>{\raggedright\arraybackslash}p{2.7cm}>{\raggedright\arraybackslash}p{0.65cm}>{\raggedright\arraybackslash}p{3cm}>{\raggedright\arraybackslash}p{1.68cm}>{\raggedright\arraybackslash}p{2.58cm}>{\raggedright\arraybackslash}p{3.38cm}X}
            \hline
            Protocol & Unif. & Onl. Comm. & Onl. Rds. & Offl. Comm. & Offl. Rds. & Approach \\
            \hline
            \cite{lu2019shamirshuffle} & $\times$ & $O(nml\log^2 m)$ & $O(\log^2 m)$ & $O(\mathcal{B}m\log^2m)$ & $O(1)$ & Switching Network \\

            \rowcolor{gray!20}
            \cite{mardi2021shuffle} & $\times$ & $O(nml\log m)$ & $O(\log m)$ & $O(\mathcal{B}m\log m)$ & $O(1)$ & Switching Network \\

            \cite{movahedi2015shuffle1, movahedi2015shuffle2} & w.p.\rlap{\textsuperscript{\text{\tiny *}}} & $O(nml\log m)$ & $O(\log m)$ & $O(\mathcal{C}m\log m)$ & $O(1)$ & Sorting Network \\

            \rowcolor{gray!20}
            \cite{lu2023rpm}  & $\approx^{\dagger}$ & $O(q n m l)$ & $O(q)$ & $O(q n m^{1.5})$ & $O(\log n)$ & Square Network \\

            \cite{laur2011shuffle} & $\sqrt{}$ & $O(2^n n^{1.5}ml\log m)$ & $O(2^n/\sqrt{n})$ & $O(1)$ & $O(1)$ & Permute-in-turn \\

            \rowcolor{gray!20}
            \cite{keller2014ORAM} & $\sqrt{}$ & $O(n^2ml\log m)$ & $O(n\log m)$ & $O(n^2m\log m)$ & $O(1)$ & Permute-in-turn \\

            Shuffle1 & $\sqrt{}$ & $O(\frac{n^2ml\log m}{\log k})$ & $O(\frac{n\log m}{\log k})$ & $O(\frac{kn^2m\log m}{\log k})$ & $O(\log\log k)$ & Permute-in-turn \\

            \rowcolor{gray!20}
            Shuffle2 and SLIDE$^{\ddagger}$ & $\sqrt{}$ & $O(nml)$ & $O(n)$ & $O(\frac{(k+l)n^2m\log m}{\log k})$ & $O(\frac{\log m}{\log k}+\log\log k)$ & Permute-in-turn \\

            \hline
        \end{tabularx}
    \begin{tablenotes}[flushleft]
        \scriptsize
        \item \textsuperscript{\textasteriskcentered} According to \cite{movahedi2015shuffle1}, this probability is $1-O(m^{-3})$.
        \item \textsuperscript{\textdagger} Statistical closeness to uniformity, i.e., the permutation distribution $\Pi_q$ satisfies $\Delta(\Pi_q, U_m) \leq O \!(m^{1 - \lfloor q/3 \rfloor \cdot \frac{1}{4}} (\log m)^{\lfloor q/3 \rfloor} \!)$.
        \item \textsuperscript{\textdaggerdbl} When $n \geq \kappa$, the offline communication complexity reduces to $O\left(\frac{k+l}{\log k}\kappa nm\log m\right)$, and the online communication complexity reduces to $O(\kappa ml)$. The $O(n)$ online round complexity also holds for \textsf{SLIDE} under dispute control; see Appendix~\ref{subsec:Round Complexity under Dispute Control} for further discussion.
\end{tablenotes}
\end{table}
\vspace*{\fill}
\end{landscape}

\paragraph{Shuffle with Guaranteed Output Delivery.}
A desirable property of shuffle protocols is guaranteed output delivery
(GOD), which ensures that honest parties obtain the correct output even
in the presence of malicious behavior. Although seemingly demanding, most
existing shuffle protocols for Shamir secret sharing already achieve
GOD~\cite{lu2019shamirshuffle, mardi2021shuffle, movahedi2015shuffle1, movahedi2015shuffle2, keller2014ORAM}.
This is because these protocols are typically built solely on MPC
multiplication and thus naturally inherit GOD when instantiated with a
Shamir secret sharing scheme that supports multiplication with GOD
(e.g.,~\cite{goyal2020guaranteed}).

In contrast, achieving GOD for uniform shuffle protocols with linear
online complexity is significantly more challenging. Linear-online
shuffles were first obtained for semi-honest additive secret sharing by
\cite{eskandarian2021clarion} and later extended to malicious security by
\cite{gao2024linearshuffle}. For efficiency, both works rely on highly
customized MPC primitives that are not designed to provide GOD, and simply
replacing the underlying secret sharing scheme does not suffice to
guarantee GOD. Consequently, whether a uniform shuffle protocol can
simultaneously achieve linear online complexity and GOD has remained open.

\paragraph{Summary.}
In this paper, we present uniform shuffle protocols for Shamir secret
sharing that substantially improve online efficiency. We first introduce
a permutation sharing technique that enables a nontrivial trade-off
between the communication cost of sharing a permutation and that of
applying it. Unlike the share translation protocol of
\cite{chase2020secret}, which uses a PRG to reduce communication in
additive secret sharing by trading communication for computation, our
construction leverages a structural property of Shamir secret sharing:
inner products of shared vectors of length $l$ can be computed with only
$O(n)$ communication and $O(nl)$ computation. This yields our first shuffle
protocol with
$O\!\left(\frac{n^2ml\log m}{\log k}\right)$ online communication and
$O\!\left(\frac{n\log m}{\log k}\right)$ online rounds. To further reduce
online complexity, we leverage the shuffle correlation technique of
\cite{gao2024linearshuffle}, obtaining our second shuffle protocol with
$O(nml)$ online communication and $O(n)$ online rounds. We then enhance
the shuffle correlation technique and our second shuffle protocol to
achieve guaranteed output delivery (GOD). This yields \textsf{SLIDE}, the
first uniform shuffle protocol that achieves both GOD and the best-known
online communication complexity of $O(nml)$, along with $O(n)$ online
rounds.

\section{Preliminaries} \label{sec:preliminary}

\subsection{Basic Notation}
We consider an MPC setting with $n$ parties
$P_1, P_2, \ldots, P_n$ participating in the computation.
Let $\mathbb{F}$ be a finite field with $|\mathbb{F}| > n$,
and denote its prime subfield by $\mathbb{F}_p$.
Let $\mathbb{K}$ be an extension of $\mathbb{F}$ such that
$|\mathbb{K}|\geq 2^\kappa$ for the statistical security parameter
$\kappa$.
Let $\{\alpha_i\}_{i=1}^n$ be $n$ distinct nonzero elements of
$\mathbb{F}$ that are publicly known to all parties.
Let $t < n/2$ denote the number of corrupted parties.

We use bold letters (e.g., $\mathbf{x}, \mathbf{y}, \mathbf{z}$) to denote vectors.
Vectors are column vectors unless stated otherwise.
We denote the $i$-th entry of a vector $\mathbf{x}$ by $x_i$
or $\mathbf{x}(i)$.
The latter notation is useful when the vector itself has a subscript,
e.g., $\mathbf{x}_1$ or $\mathbf{x}_2$.
Bold uppercase letters
(e.g., $\mathbf{A}, \mathbf{B}, \mathbf{C}$) denote matrices,
and the entry in the $i$-th row and the $j$-th column of $\mathbf{A}$
is denoted by $a_{i,j}$ or $\mathbf{A}(i,j)$.

Define $[m] = \{1, 2, \ldots, m\}$ for any positive integer $m$.
For a vector $\mathbf{x} = (x_i)_{i\in [m]}$ of length $m$
and an $m$-permutation $\pi: [m]\to [m]$, applying $\pi$ to
$\mathbf{x}$ yields
\[
\pi(\mathbf{x}) \coloneqq (x_{\pi(1)}, x_{\pi(2)}, \ldots, x_{\pi(m)}).
\]
The permutation matrix of a permutation $\pi$ is defined as
$\mathbf{P}_\pi \coloneqq (p_{i,j})_{i,j\in [m]}$, where
\[
p_{i,j} =
\begin{cases}
1 & \text{if } j=\pi(i),\\
0 & \text{otherwise}.
\end{cases}
\]
With $i$ as the row index and $j$ as the column index, we have
\[
\pi(\mathbf{x}) = \mathbf{P}_\pi\cdot \mathbf{x}.
\]

We denote the Shamir secret sharing of a value \(s\)
among all parties by \(\llbracket s \rrbracket\).
This means that there exists a degree-\(t\) polynomial
\(f \in \mathbb{F}[X]\), where \(t\) is the number of corrupted parties,
such that \(f(0)=s\) and \(f(\alpha_i)\) is known exclusively to \(P_i\).
For a secret-shared vector \(\llbracket \mathbf{v} \rrbracket\) or
matrix \(\llbracket \mathbf{A} \rrbracket\), every entry is
secret-shared individually. For notational simplicity, we write a shared
vector as
\[
\llbracket \mathbf{v} \rrbracket
=
(\llbracket v_1 \rrbracket, \llbracket v_2 \rrbracket, \ldots,
\llbracket v_m \rrbracket)^\top
=
\llbracket v_1, v_2, \ldots, v_m \rrbracket^\top,
\]
with the same convention applying to shared matrices.

When invoking a sub-protocol, three types of parameters may be involved.
For instance, an invocation written as
\[
\mathrm{Prot}(a, \llbracket b \rrbracket, P_i: c)
\]
denotes invoking protocol \(\mathrm{Prot}\) with a constant \(a\),
a secret-shared value \(\llbracket b \rrbracket\), and a private input
\(c\) from party \(P_i\).

We assume that the items to be shuffled consist of $m$ field elements or
$m$ vectors, where $m$ is a power of 2. The main reason for this
assumption is to leverage a technique proposed by
\cite{chase2020secret}, which uses a Bene\v{s} network
\cite{benevs1964permutation} for permutation decomposition.
The Bene\v{s} network used in this decomposition is designed for
permutations whose size is a power of 2. Since a shuffle for an array
of size $m$ can be readily used to shuffle an array of smaller size
$m^\prime \leq m$ (as discussed in \cite{gao2024linearshuffle}), this
assumption is without loss of generality.

\subsection{Security Model}
We consider a malicious adversary capable of corrupting at most $t < n/2$ parties.
    We assume the existence of secure pairwise channels and a secure broadcast channel.
    For simplicity, we measure communication complexity in terms of
    the number of field elements transmitted,
    implicitly assuming $\log |\mathbb{K}| = O(1)$.
    Thus, broadcasting a field element in either $\mathbb{F}$ or $\mathbb{K}$
incurs $O(n)$ communication.

For clarity, our primary objective is to construct shuffle protocols
that achieve security with abort, with unconditional security against
static adversaries under an honest majority.
A static adversary selects and fixes the set of parties to corrupt before
the protocol’s initiation. Unconditional security means our construction
does not rely on any computational hardness assumptions,
i.e., the adversary may be computationally unbounded.
    Specifically, as long as the number of corrupted parties is less
    than $n/2$, our construction remains statistically secure:
if corrupted parties deviate from the protocol,
honest parties will abort with overwhelming
probability $p>1-C\cdot 2^{-\kappa}$,
where $C$ denotes the number of MPC operations
(i.e., the size of the computation circuit).
Such a strong guarantee is known to be achievable only in the
honest-majority setting \cite{ben2019completeness, chaum1988multiparty}.

In Section \ref{sec:GOD},
    we extend our constructions to guarantee output delivery (GOD).
    GOD is a stronger security guarantee stating that honest parties
    always obtain correct results regardless of adversarial behavior.
This is achieved by replacing all security-with-abort MPC primitives
    with their GOD counterparts \cite{goyal2020guaranteed},
    applying the dispute control framework of \cite{beerliova2006dispute}
    to identify and exclude corrupted parties,
    and using the authentication tag mechanism of \cite{Ben-Sasson2012miniMPC} to
    recover a correct intermediate state for honest parties.

To preserve the continuity of the main exposition, the detailed security proofs are deferred to Appendix~\ref{sec:security proofs}.

\subsection{Primitives for Shamir Secret Sharing}
Our constructions
    require only a high-level interface to
    Shamir secret sharing, which we formalize through
    an ideal MPC functionality $\mathcal{F}_{\mathrm{Shamir}}$
    supporting the following commands:
\begin{itemize}
    \item $\llbracket s \rrbracket \gets \Pi_{\mathrm{share}}(P_i: s)$:
    share a secret from $P_i$ among all parties.
    \item $s\gets \Pi_{\mathrm{open}}(\llbracket s \rrbracket)$:
    publicly open a secret $s$ to all parties.
    \item $s\gets \Pi_{\mathrm{open}}(\llbracket s \rrbracket, P_i)$:
    open $s$ to party $P_i$.
    \item $\llbracket s \rrbracket \gets \Pi_{\mathrm{add}}(\llbracket a \rrbracket, \llbracket b \rrbracket)$:
    return $s=a+b$. This protocol requires no communication.
    \item $\llbracket s \rrbracket \gets \Pi_{\mathrm{mul}}(\llbracket a \rrbracket, \llbracket b \rrbracket)$:
    return $s=ab$ with rerandomization.
    \item $\llbracket s\rrbracket \gets \Pi_{\mathrm{inner}}(\llbracket \mathbf{a} \rrbracket, \llbracket \mathbf{b}\rrbracket)$:
    return the inner product $s=\mathbf{a}^\top\cdot \mathbf{b}$ with rerandomization.
    \item $\llbracket r \rrbracket \gets \Pi_{\mathrm{rand}}(\mathbb{J})$:
    draw $r\overset{\$}{\gets} \mathbb{J}$ uniformly, where $\mathbb{J}$ is
    either $\mathbb{F}_p$, $\mathbb{F}$, or $\mathbb{K}$.
    \item $\lambda\gets \Pi_{\mathrm{challenge}}(\mathbb{K})$:
    draw a public challenge $\lambda$ uniformly from the designated field
    $\mathbb{K}$.
\end{itemize}

The functionality $\mathcal{F}_{\mathrm{Shamir}}$ models Shamir secret
sharing at a high level. It provides a modular interface for protocol
design without committing to a specific implementation.
All operations are assumed to be executed correctly, and any misbehavior
causes the protocol to abort.
In practice, $\mathcal{F}_{\mathrm{Shamir}}$ can be instantiated using
standard MPC protocols based on Shamir secret sharing, in either the
semi-honest or the malicious setting (e.g.,~\cite{goyal2020guaranteed}).
By replacing calls to $\mathcal{F}_{\mathrm{Shamir}}$ with the
corresponding concrete protocols, correctness follows directly, and
security is preserved via standard composition theorems.

Note that $\Pi_\mathrm{share}$, $\Pi_\mathrm{open}$,
$\Pi_\mathrm{mul}$, $\Pi_\mathrm{rand}$, and
$\Pi_\mathrm{challenge}$ can be implemented with $O(n)$ communication
and $O(n)$ computation.
The inner-product primitive $\Pi_{\mathrm{inner}}$ requires $O(nm)$
computation but only $O(n)$ communication when invoked on vectors of
length $m$.
In our later constructions, we leverage this asymmetry to trade
computation for communication and thereby improve the overall complexity.

A concrete realization of these primitives with guaranteed output
delivery (GOD) is provided by~\cite{goyal2020guaranteed}.
In practice, implementations of $\mathcal{F}_{\mathrm{Shamir}}$ rely on
batching to achieve amortized linear communication.
This amortized complexity is fully compatible with our protocols, since
we do not require immediate correctness checks for each multiplication.
We also note that, for efficiency reasons, the concrete implementation
of $\Pi_\mathrm{challenge}$ produces a distribution that is only
``close to'' uniform and incurs $O(n)$ communication overhead
\cite{goyal2020guaranteed}.
For simplicity, we abstract away this detail and assume perfect
uniformity; this abstraction does not affect the security of our protocol.

\subsection{Generic Approach for GOD}
    We enhance our constructions to guarantee output delivery (GOD)
    by leveraging the primitives and framework
    of \cite{goyal2020guaranteed}, which provides a generic approach
    for realizing GOD in the honest-majority setting via
    Shamir secret sharing.
    Their framework builds on the
    dispute control technique of \cite{beerliova2006dispute},
    a strategy for identifying corrupted parties.

    \paragraph{Dispute Control.}
    The core idea of the dispute control technique is to maintain two public sets:

(1) $\mathrm{Corr}$, containing parties already identified as corrupted.

(2) $\mathrm{Disp}$, containing the pairs of parties ${P_i,P_j}$ currently in dispute
(meaning at least one of $P_i$ and $P_j$ must be corrupted).

        It follows that if a party is disputed with at least $(t+1)$ parties,
        it must be corrupted.
        Additionally, any party identified as corrupted is automatically placed
        in dispute with all remaining parties.
        The primitives proposed by \cite{goyal2020guaranteed} are designed
        such that all honest parties consistently agree on the state of
        these two sets, and innocent parties are never wrongfully added
        to $\mathrm{Corr}$.

    Parties proceed with the computation.
    Whenever a correctness check fails,
    the parties revert to the beginning of the current segment and
    update $\mathrm{Corr}$ or $\mathrm{Disp}$.
    As every failure enlarges one of the two sets,
    and as the number of potential disputes is $O(n^2)$,
    the computation will be rolled back at most $O(n^2)$ times, which is
sufficient to identify all corrupted parties.

    Communication between disputed parties is restricted: they never send messages directly to each other.
        If $P_i$ needs to share a secret, it simply sets the share intended for $P_j$ to zero.
        If $P_i$ must send a non-sharing message to $P_j$,
        it routes the message through a relay $P_{i \leftrightarrow j}$,
        who is the lowest-index party not in dispute with either of them.
        Subsequently, any inconsistency in the relayed message forces this relay to side
        with one party, creating a new dispute and ensuring
        eventual corruption identification.

    When a corrupted party is finally excluded from computation,
    its shares become unusable.
    This poses a challenge for multiplication computations,
        which need at least $2t+1$ valid shares to proceed.
        \cite{goyal2020guaranteed} proposes a primitive
        dubbed ``small surgery'' to enable computation
        to continue by nullifying the shares held by corrupted parties.
        Note that this does not impact our constructions,
        as our constructions treat $\Pi_{\mathrm{mul}}$ and
        related primitives as black boxes.
        Hence, these details are abstracted away by our
        ideal-functional assumption.

    \paragraph{Circuit Division.}
    For efficiency optimization, the arithmetic circuit $\mathcal{C}$ is divided into $O(n^2)$ segments.
        Following the computation of each segment,
        a correctness check is performed to verify that all computations up
        to that point are correct. If the check passes, parties additionally
        compute authentication tags \cite{Ben-Sasson2012miniMPC},
        which act as commitments to the shares produced in that segment.
        Broadly speaking, the purpose of these authentication tags is to prevent corrupted parties
        from lying about their shares obtained in earlier segments.
        The rationale is that if a party falsifies its shares for
        the current segment, parties can simply recompute that segment.
        However, if a party falsifies shares from a previous segment,
        rewinding the computation to that segment is impractical due to
        efficiency constraints.
        In such cases, the authentication tags will help identify the
        lying party and add it to $\mathrm{Corr}$.
        This ensures the computation never reverts to previous segments.
        Since there will be at most $O(n^2)$ rollbacks and each segment has size
$O(|C|/n^2)$, the overall computation remains $O(|C|)$,
linear in the size of the circuit.

\section{Permutation Sharing Protocol and Shuffle Protocol}\label{sec:permutation}

In this section, we begin by constructing a permutation sharing protocol
that enables a designated party $P_w$ to efficiently share a secret
permutation $\pi$, which can subsequently be applied to any vector of
compatible size. Given such a permutation sharing mechanism, a shuffle
protocol follows directly via the standard permute-in-turn paradigm:
in the offline phase, each party shares a secret, uniformly random
$m$-permutation; in the online phase, all parties apply the shared
permutations sequentially. The uniformity of the final shuffle is
guaranteed by the fact that, even if only one permutation remains secret
and uniform, the composite permutation is uniform.

Consider permuting $m$ secret-shared vectors of length $l$.
There are two straightforward approaches to implementing such a
permutation sharing scheme:
    \begin{enumerate}[noitemsep, topsep=2pt, partopsep=2pt]
        \item Share an $m\times m$ permutation matrix among parties.
            Sharing such a permutation incurs $O(nm^2)$ communication,
            while applying it incurs $O(nml)$ communication, $O(nm^2l)$ computation
            and $O(1)$ rounds.
        \item Share $m\log m$ bits among parties.
            Each bit corresponds to a control bit
            in the permutation network \cite{keller2014ORAM}.
            Sharing a permutation hence incurs $O(nm\log m)$ communication,
            while applying it incurs $O(nml\log m)$ communication,
            $O(nml\log m)$ computation and $O(\log m)$ rounds.
    \end{enumerate}
The first approach achieves lower communication overhead for permutation
application, while the second offers better overall communication
complexity. The key limitation of the first approach is its
\(O(nm^2)\) communication and computation overhead for permutation
sharing, which is prohibitive for large $m$. The key drawbacks of the
second approach are its \(O(\log m)\) round complexity and the fact that
permutation application is no more efficient than permutation sharing
itself, meaning that the protocol does not effectively leverage the
offline phase.

Our construction unifies these extremes. We observe that an
$m$-permutation can be represented using
$\frac{m\log m}{k\log k}$ many $k\times k$ permutation matrices.
Permutation application then reduces to multiplying these $k\times k$
matrices with $k$-dimensional vectors in the correct order.
Sharing a permutation via our approach thus incurs
$O(\frac{knm\log m}{\log k})$ communication, while applying the
permutation incurs $O(\frac{nml\log m}{\log k})$ communication and
$O(\frac{nmkl\log m}{\log k})$ computation.

To realize this, we require a protocol for sharing well-formed
permutation matrices. However, verifying whether an arbitrarily shared
matrix is a well-formed permutation matrix is generally nontrivial
(see the discussion in Appendix~\ref{sec:Generating Permutation Matrix}).
To address this challenge, we first enable a party to share $k$ distinct
one-hot vectors. By verifying that these $k$ one-hot vectors are pairwise
distinct, the parties can confirm the well-formedness of the resulting
permutation matrix.

We proceed as follows. We first formalize the definition of a shared
permutation and illustrate how parties use it to perform permutation
operations. We next present our construction for sharing one-hot vectors
as a key building block for generating shared permutation matrices.
We then describe our permutation sharing protocol. Finally, we present
the shuffle protocol built from this permutation sharing protocol,
followed by a comprehensive complexity analysis.

\subsection{Shared Permutation}

\begin{definition}[$k$-Shared $m$-Permutation]\label{definition: shared permutation}
    We say that an $m$-permutation $\pi$ is $k$-shared, denoted by
    $\ll \pi\rr_k$, if
    $$ \ll \pi \rr_k \coloneqq (k, s, \{\{\ll \mathbf{P}_{\pi^\prime_{i, j}}\rr, T_{i, j}\}_{j\in [m/k]}\}_{i\in [s]}),$$
    where
    \begin{enumerate}[noitemsep, topsep=2pt, partopsep=2pt]
        \item Each $\pi_{i, j}^\prime$ is a $k$-permutation.
        \item $s=O(\frac{\log m}{\log k})$. This bound is tight; see
        Appendix~\ref{subsec:Definition of Shared Permutation}.
        \item $\mathbf{P}_{\pi^\prime_{i, j}}$ denotes the permutation matrix
        corresponding to \(\pi^\prime_{i,j}\), which is secret-shared among parties.
        \item $T_{i, j}\subseteq [m]$ is a subset of size $k$, referred to as
        the ``task'' of $\pi^\prime_{i, j}$. For a fixed $i$, the sets
        $T_{i, j}$ are disjoint.
        \item $\pi \! =\! \pi_{s, 1}\circ\pi_{s, 2}\circ\cdots\circ\pi_{s, m/k}
        \circ\pi_{s-1, 1}\circ\cdots \circ\pi_{1, m/k}$,
        where $\circ$ denotes permutation composition. Here, \(\pi_{i,j}\)
        is the $m$-permutation induced by applying \(\pi^\prime_{i,j}\) to
        the entries specified by \(T_{i,j}\).
    \end{enumerate}

    For simplicity, we further require that $m$ and $k$ are powers of 2.
    When $k$ is irrelevant to the context, we abbreviate $\ll\pi\rr_k$ as
    $\ll\pi\rr$.
\end{definition}

Once $\ll \pi\rr$ is shared, the parties can apply $\pi$ to any number
of $m$-dimensional vectors. This is done by iteratively replacing the
relevant entries of the $m$-vector with the product of the corresponding
shared $k\times k$ permutation matrix and the extracted subvector.
For instance, when $k=m$, the shared permutation reduces to a single
secret-shared $m\times m$ permutation matrix, and multiplying it by the
vector directly yields the permuted vector. This process is formalized
in Algorithm~\ref{alg:Permute}.

\begin{algorithm}
    \caption{$\ll \pi(\mathbf{v})\rr \gets \mathrm{Permute}(\ll\pi\rr, \ll \mathbf{v}\rr)$}\label{alg:Permute}
    \begin{algorithmic}
    \Require $\pi$ is a shared $m$-permutation and $\mathbf{v}$ is of length $m$.
    \Ensure Return $\ll \pi(\mathbf{v})\rr$.\;
    \State Parse
    $\ll \pi\rr$ as $(k, s, \{\{\ll \mathbf{P}_{\pi^\prime_{i, j}}\rr, T_{i, j}\}_{j\in [m/k]}\}_{i\in [s]})$\;
    \For{$i=1$ to $s$}
        \For{$j=1$ to $m/k$} \textbf{parallel}
            \State $\ll \mathbf{u}\rr \gets \ll \mathbf{v}(l)\rr_{l\in T_{i, j}}$\;
            \State $\ll \mathbf{u}^\prime\rr \gets \ll \mathbf{P}_{\pi'_{i,j}}\rr\cdot \ll \mathbf{u}\rr$\; \Comment{$k$ invocations of $\Pi_\mathrm{inner}$.}
            \State $\ll \mathbf{v}(l)\rr\gets \ll \mathbf{u}'(l)\rr$ for each index $l\in T_{i, j}$\;
        \EndFor
    \EndFor
    \State Return $\ll \mathbf{v}\rr$.\;
    \end{algorithmic}
\end{algorithm}

\subsection{Sharing a One-Hot Vector}
A one-hot vector is a vector with exactly one nonzero entry, whose value is $1$.
We define the ``index'' of a one-hot vector as the position of the $1$.
Sharing one-hot vectors is the first step toward constructing a shared
permutation matrix, as a $k\times k$ permutation matrix consists of $k$
pairwise distinct one-hot vectors.

Suppose party $P_w$ intends to share a $k$-dimensional one-hot vector
with index $\mathrm{ind}$, where $k=2^d$. First, $P_w$ decomposes the
index into $\log k=d$ bits $b_1, b_2, ..., b_d$ such that
$$
 \mathrm{ind} = \sum_{i=1}^{d} 2^{i-1}b_i.
 $$
Next, it secret-shares each $b_i$ among all parties. To check that all
shared values are binary, the parties generate a random challenge
\(\lambda \in \mathbb{K}\) and verify that
$$
0 = \sum_{i=1}^{d} \lambda^{i-1}\ll b_i\rr\ll 1-b_i\rr.
$$
This check evaluates a polynomial \(F \in \mathbb{K}[X]\) of degree at
most \(d-1\) at a random point \(\lambda \in \mathbb{K}\).
    Since a polynomial of degree \(d-1\) has at most \(d-1\) roots in \(\mathbb{K}\),
    the probability that the check detects a non-binary \(b_i\) is overwhelming:
    $$ \Pr[F(\lambda)\neq 0\mid \exists i \text{ s.t. } b_i(1-b_i)\neq 0] \geq 1-\frac{d-1}{|\mathbb{K}|}\geq 1-\frac{d-1}{2^\kappa}. $$
    This binary checking protocol is described in Algorithm \ref{alg:BinCheck}.

    \begin{algorithm}
        \caption{$\mathrm{BinCheck}(\ll b_1\rr, ..., \ll b_d\rr)$}\label{alg:BinCheck}
        \begin{algorithmic}
        \Ensure Abort if any $b_i$ is not in $\{0, 1\}$.
        \State $\lambda\gets\Pi_\mathrm{challenge}(\mathbb{K})$\;
        \State $\ll u\rr\gets \sum_{i=1}^d \lambda^{i-1}\ll b_i\rr\ll 1-b_i\rr$\;\Comment{One call to $\Pi_\mathrm{inner}$.}
        \State $u\gets\Pi_\mathrm{open}(\ll u\rr)$\;
        \State All parties abort if $u\neq 0$.\;
        \end{algorithmic}
    \end{algorithm}

If the binary check passes, the parties convert the shared bits into a
one-hot vector. This process of converting $\log k$ bits into a one-hot
vector of length $k$ is known as demultiplexing. In MPC, this can be done
using the demux protocol proposed by \cite{launchbury2012demux}, which
computes the desired one-hot vector with $O(k)$ MPC multiplications in
$O(\log \log k)$ rounds. The protocol works recursively: it splits the
$d$ bits into two subsets, generates a one-hot vector for each subset via
recursive calls, and then computes the tensor product of these two
vectors to obtain the final one-hot vector with the correct index.
For completeness, we present this demux protocol in
Algorithm~\ref{alg:Demux}.

Algorithm~\ref{alg:OneHot} describes our $\mathrm{OneHot}$ protocol,
which enables party $P_w$ to share a one-hot vector. Notably, the parties
can verify all shared bits for multiple one-hot vectors with a single call
to \(\mathrm{BinCheck}\) by batching all their bit shares, rather than
performing separate checks for each one-hot vector.

\begin{algorithm}
    \caption{$(\ll s_1\rr, ..., \ll s_{2^d}\rr) \gets \mathrm{Demux}(\ll b_1\rr, ..., \ll b_{d}\rr)$}\label{alg:Demux}
    \begin{algorithmic}
    \Ensure Output a one-hot vector with index $\sum_{i=1}^{d} 2^{i-1} b_i + 1$.
    \If{$d=1$}
        \State Return $(1-\llbracket b_1\rrbracket, \llbracket b_1\rrbracket)$\;
    \EndIf
    \State $t \gets \lfloor d/2 \rfloor$\;
    \State $\ll u_1\rr, \ll u_2\rr, ..., \ll u_{2^t}\rr \gets \mathrm{Demux}(\ll b_1\rr, ..., \ll b_{t}\rr)$\;
    \State $\ll v_{1}\rr, \ll v_{2}\rr, ..., \ll v_{2^{d-t}}\rr \gets \mathrm{Demux}(\ll b_{t+1}\rr, ..., \ll b_{d}\rr)$\;\Comment{Parallel recursive calls.}
    \For{$i=1$ to $2^t$} \textbf{parallel}
        \For{$j=1$ to $2^{d-t}$} \textbf{parallel}
            \State $\ll s_{i+2^t(j-1)}\rr\gets\Pi_\mathrm{mul}(\ll u_i \rr, \ll v_j \rr)$\;
        \EndFor
    \EndFor
    \State Return $\ll s_1, ..., s_{2^d}\rr$\;
    \end{algorithmic}
\end{algorithm}

\begin{algorithm}
    \caption{$(\ll s_1\rr, ..., \ll s_{k}\rr) \gets \mathrm{OneHot}(k, P_w: \mathrm{ind})$}\label{alg:OneHot}
    \begin{algorithmic}
    \Require $k=2^d$ is a power of 2.
    \Ensure Output a one-hot vector with index $\mathrm{ind}$.
    \State $P_w$ computes bits $b_1, ..., b_d$ such that $\mathrm{ind} = \sum_{i=1}^d 2^{i-1}b_i+1$.\;
    \State $P_w$ shares $\ll b_1\rr$, ..., $\ll b_d\rr$.\;
    \State $\mathrm{BinCheck}(\ll b_1\rr, ..., \ll b_d\rr)$\;
    \State Return $\mathrm{Demux}(\ll b_1\rr, ..., \ll b_d\rr)$.\;
    \end{algorithmic}
\end{algorithm}

\subsection{Sharing a Permutation Matrix}
A $k\times k$ permutation matrix consists of $k$ one-hot row vectors with
distinct indices. Thus, party $P_w$ can share such a matrix by invoking
the $\mathrm{OneHot}$ protocol $k$ times and assembling the resulting row
shares into a matrix $\mathbf{P}$. While each row then contains exactly
one~$1$, a corrupted $P_w$ may render the matrix ill-formed by mapping
multiple rows to the same column. Therefore, the parties must additionally
verify that every column sums to~$1$.

Opening each column sum individually requires $k$ openings. To batch the
verification, the parties instead sample a random challenge
$\lambda\gets \Pi_\mathrm{challenge}(\mathbb{K})$ and check
\[
 \sum_{i=1}^k \left(\lambda^{i-1}\left(1 - \sum_{j=1}^k \mathbf{P}(j, i) \right)\right) = 0.
\]
Since a nonzero polynomial of degree $k-1$ has at most $k-1$ roots in
\(\mathbb{K}\), the equality fails to hold with probability at least
$1-\frac{k-1}{|\mathbb{K}|}\geq 1-\frac{k-1}{2^\kappa}$
if any column has a sum different from~$1$.

Algorithms~\ref{alg:PermMat} and~\ref{alg:PermCheck} detail the sharing
and verification procedures. The entire process of
Algorithm~\ref{alg:PermCheck} consists of one call to
$\Pi_{\mathrm{challenge}}$ and one call to $\Pi_\mathrm{open}$.

\begin{algorithm}
    \caption{$\ll \mathbf{P}\rr \gets \mathrm{PermMat}(P_w: \pi)$}\label{alg:PermMat}
    \begin{algorithmic}
    \Require $\pi$ is a $k$-permutation, where $k$ is a power of 2.
    \Ensure Output a permutation matrix $\mathbf{P}$.
    \For{$i = 1$ to $k$} \textbf{parallel}
        \State $\mathbf{v}_i \gets \mathrm{OneHot}(k, P_w: \pi(i))$\;
    \EndFor
    \State $\ll \mathbf{P}\rr \coloneqq \ll v_{i, j}\rr_{i, j\in [k]}$\;
    \State $\mathrm{PermCheck}(\ll \mathbf{P}\rr)$\;
    \State Return $\ll \mathbf{P}\rr$.\;
    \end{algorithmic}
\end{algorithm}

    \begin{algorithm}
        \caption{$\mathrm{PermCheck}(\ll\mathbf{P}\rr)$}\label{alg:PermCheck}
        \begin{algorithmic}
        \Require $\mathbf{P}$ is of size $k\times k$, with each row a one-hot vector.
        \Ensure Abort with high probability if $\mathbf{P}$ is not a permutation matrix.
        \State $\ll \mathrm{sum}\rr\gets \ll 0 \rr$\;
        \State $\lambda\gets \Pi_\mathrm{challenge}(\mathbb{K})$\;
        \For{$i=1$ to $k$}\Comment{Enumerate columns.}
            \State $\ll \mathrm{sum}\rr \gets \ll\mathrm{sum}\rr + \lambda^{i-1}(1-\sum_{j = 1}^{k} \ll \mathbf{P}(j, i)\rr)$\;
        \EndFor
        \State $\mathrm{sum} \gets \Pi_\mathrm{open}(\ll \mathrm{sum}\rr)$\;
        \State Abort if $\mathrm{sum}\neq 0$.\;
        \end{algorithmic}
    \end{algorithm}

\subsection{Decomposing and Sharing a Permutation}

The above protocols allow the parties to share a $k\times k$ permutation
matrix. To share an $m$-permutation for large $m$, we first represent it
using many $k$-permutations. This is achieved using the Bene\v{s} network.

The Bene\v{s} network, due to the seminal work of \cite{benevs1964permutation},
is a switching network capable of representing any $m$-permutation with
$O(\log m)$ depth. By partitioning the network into layers and extracting
intermediate permutations, one obtains a natural decomposition of the
original permutation. \cite{chase2020secret} introduced a permutation
decomposition technique based on this reformulation, decomposing an
$m$-permutation into $O(\frac{m\log m}{k\log k})$ smaller permutations,
each of size~$k$. The classic Bene\v{s} network requires $m$ to be a power
of~2. Although there exist generalizations to arbitrary~$m$, they lose some
important properties, such as the independence of the small decomposed
permutations, which is crucial for keeping the round complexity at
$O(\log m)$.

To formalize this process, let
$$ (\pi_1, \pi_2, ..., \pi_s) \gets \mathrm{Decompose}(\pi, k) $$
denote the process of decomposing $\pi$ such that
$$\pi = \pi_s\circ\pi_{s-1}\circ\cdots \circ\pi_1,$$
where each $\pi_i$ is an $m$-permutation. According to
\cite{chase2020secret}, the concrete number of layers is
$s = (2\log m - 1)/\log k = O(\frac{\log m}{\log k})$.
Each $\pi_i$ can be further decomposed as
$\{\pi_{i,j}\}_{j\in[m/k]}$, where each $\pi_{i,j}$ is an
$m$-permutation that affects only $k$ entries. Thus,
$$\pi = \pi_{s,1}\circ\pi_{s, 2}\circ\cdots\circ\pi_{s, m/k}\circ\pi_{s-1, 1}\circ\cdots \circ\pi_{1, m/k}. $$

Define the ``task'' of each $\pi_{i, j}$ as
$$ \mathrm{Task}(m, k, i, j) = \{x\in [m]\mid \pi_{i, j}\text{ affects the }x\text{-th entry}\}. $$
Let $\pi'_{i,j}$ be the compressed (i.e., $k$-permutation) version of
$\pi_{i,j}$ acting only on its task. The decomposition guarantees that
each task has size~$k$.

A key property of this decomposition is its task invariance:
the task of \(\pi_{i,j}\) is independent of the original permutation \(\pi\).
This allows all parties to agree on the tasks in advance
(i.e., which entries to permute) without learning anything about the
specific permutation.

\begin{example}
Consider decomposing an $8$-permutation $\pi$ with $k=4$. Suppose the
decomposition yields
$$ \pi_1 = \left(\begin{matrix}
    1 & 2 & 3 & 4 & 5 & 6 & 7 & 8\\
    5 & 4 & 3 & 6 & 7 & 8 & 1 & 2
\end{matrix}\right).$$
It can be further decomposed into two disjoint permutations:
$$ \pi_{1,1} = \left(\begin{matrix}
    1 & 2 & 3 & 4 & 5 & 6 & 7 & 8\\
    5 & 2 & 3 & 4 & 7 & 6 & 1 & 8
\end{matrix}\right),$$
$$ \pi_{1,2} = \left(\begin{matrix}
    1 & 2 & 3 & 4 & 5 & 6 & 7 & 8\\
    1 & 4 & 3 & 6 & 5 & 8 & 7 & 2
\end{matrix}\right).$$
$\pi_{1,1}$ and $\pi_{1,2}$ are disjoint $8$-permutations, in the sense
that they act on disjoint fixed subsets (i.e., their tasks) of
$\{1,..., 8\}$:
$$ \begin{aligned}
    \mathrm{Task}(8, 4, 1, 1) &= \{1, 3, 5, 7\},\\
    \mathrm{Task}(8, 4, 1, 2) &= \{2, 4, 6, 8\}.
\end{aligned}$$
Although $\pi_{1,1}$ may vary with $\pi$, its task does not, due to the
task-invariant property. We may compress these two permutations into
$\pi^\prime_{1,1}$ and $\pi^\prime_{1,2}$:
$$
\pi^\prime_{1,1} = \left(\begin{matrix}
    1 & 2 & 3 & 4\\
    3 & 2 & 4 & 1
\end{matrix}\right),\
\pi^\prime_{1,2} = \left(\begin{matrix}
    1 & 2 & 3 & 4\\
    2 & 3 & 4 & 1
\end{matrix}\right).
$$

Thus, $\pi_1$ is represented by two disjoint $4$-permutations
$\pi'_{1,1}$ and $\pi'_{1,2}$.
\end{example}

Our permutation sharing protocol is given in Algorithm~\ref{alg:SharePerm}.
The parties first agree on the parameter $k$ and locally compute a
sequence of tasks $\{\mathrm{Task}(m,k,i,j)\}$. The dealer $P_w$ chooses
a permutation $\pi$ and locally decomposes $\pi$ into the small
permutations $\{\pi'_{i,j}\}$. Then $P_w$ shares each small permutation
as a permutation matrix. Since the permutations $\{\pi_{i,j}\}$ for a
fixed $i$ are disjoint, this yields a valid $k$-shared $m$-permutation.

    \begin{algorithm}
        \caption{$\ll \pi\rr_k \gets \mathrm{SharePerm}(k, P_w: \pi)$}\label{alg:SharePerm}
        \begin{algorithmic}
        \Require $\pi$ is an $m$-permutation.
        \Ensure Return a $k$-shared $m$-permutation $\ll \pi\rr_k$.\;
        \State Party $P_w$ lets $ (\pi_1, ..., \pi_{s})\gets \mathrm{Decompose}(\pi, k)$.\;
        \State Party $P_w$ further decomposes each $\pi_i$ into $\pi_{i, 1}, ..., \pi_{i, m/k}$.\;
        \For{$i=1$ to $s$} \textbf{parallel}
            \For{$j=1$ to $[m/k]$} \textbf{parallel}
                \State $T_{i, j}\gets \mathrm{Task}(m, k, i, j)$\;
                \State $\ll\mathbf{P}_{\pi'_{i, j}}\rr\gets \mathrm{PermMat}(P_w: \pi^\prime_{i, j})$\;
            \EndFor
        \EndFor
        \State Return $\ll \pi\rr_k\coloneqq (k, s, \{\{\ll \mathbf{P}_{\pi'_{i, j}}\rr, T_{i, j}\}_{j\in [m/k]}\}_{i\in [s]})$.\;
        \end{algorithmic}
    \end{algorithm}

\subsection{Shuffle Protocol}
Our first shuffle protocol consists of exactly $n$ calls to the
permutation protocol developed above, following a permute-in-turn
mechanism where parties apply their permutations sequentially.
The two phases of the protocol are formally presented in
Algorithms~\ref{alg:FirstShuffleOff} and~\ref{alg:FirstShuffleOn}.
In the offline phase, the parties prepare $n$ shared permutations, one
from each party. In the online phase, the parties apply these permutations
to the incoming data by invoking $\mathrm{Permute}$.

    \begin{algorithm}
        \caption{$\{\ll \pi_i\rr_k\}_{i\in [n]} \gets \mathrm{Shuffle1}_\mathrm{off}(k, m)$}\label{alg:FirstShuffleOff}
        \begin{algorithmic}
        \Require Both $m$ and $k$ are power of 2.
        \For{$w=1$ to $n$}
            \State Party $P_w$ draws an $m$-permutation $\pi_w$ uniformly.\;
            \State $\mathrm{SharePerm}(k, P_w: \pi_w)$\;
        \EndFor
        \State Return $\{\ll \pi_i\rr_k\}_{i\in [n]}$.\;
        \end{algorithmic}
    \end{algorithm}

\begin{algorithm}
    \caption{$\ll \mathbf{v}^\prime_1, \mathbf{v}^\prime_2, ..., \mathbf{v}^\prime_l\rr \gets \mathrm{Shuffle1}_\mathrm{on}(\ll \mathbf{v}_1, \mathbf{v}_2, ..., \mathbf{v}_l\rr)$}\label{alg:FirstShuffleOn}
    \begin{algorithmic}
    \Require Each $\mathbf{v}_i$ is of length $m$, where $m$ is a power of 2.
    \Ensure Return the shuffled shares $\ll \pi(\mathbf{v}_i)\rr_{i\in [l]}$, with a secret uniform $\pi$.\;
    \State Fetch fresh $\{\ll\pi_i\rr\}_{i\in [n]}$ generated by $\mathrm{Shuffle1}_\mathrm{off}$.\;
    \For{$w=1$ to $n$}
        \For{$i=1$ to $l$} \textbf{parallel}
            \State $\ll \mathbf{v}_i\rr\gets \mathrm{Permute}(\ll\pi_w\rr, \ll \mathbf{v}_i\rr)$\;
        \EndFor
    \EndFor
    \State Return $\ll \mathbf{v}^\prime_1, \mathbf{v}^\prime_2, ..., \mathbf{v}^\prime_l\rr$.\;
    \end{algorithmic}
\end{algorithm}

\paragraph{Security of \textsf{Shuffle1}.}
\rev{The security of \textsf{Shuffle1} follows from three observations. First,
the final permutation is the composition of the permutations chosen by all
parties. Hence, as long as one honest party chooses its permutation
uniformly at random and keeps it hidden, the composed permutation is also
uniformly random, regardless of the permutations chosen by corrupted
parties. Second, the privacy of each honest party's permutation is protected
by Shamir secret sharing: the small permutation matrices in
\(\ll \pi \rr_k\) are secret-shared, while the public task sets in the
decomposition are independent of the actual permutation. Thus, the task
structure reveals only where local \(k\)-permutations are applied, but not
which local permutations are used. Third, maliciously chosen permutations
are forced to be well-formed. The \textsf{BinCheck} protocol ensures that
the shared row encodings are binary, and \textsf{PermCheck} ensures that the
resulting one-hot rows form a valid permutation matrix. Therefore, a
malicious dealer can only make the protocol abort with overwhelming
probability or contribute a valid permutation. Detailed security proofs are
provided in Appendix~\ref{sec:security proofs}.}

\subsection{Complexity Analysis}\label{subsec:complexity of sec4}
The following theorems analyze the complexity of the key protocols
introduced in this section. Their proofs, together with the complexity
analyses of the remaining protocols referenced above, are provided in
Appendix~\ref{subsec:proof of sec4}. The complexity of the shuffle
protocols, with a clear separation between offline and online costs, is
summarized in Theorem~\ref{theorem:Complexity of First Shuffle}.

\begin{theorem}\label{theorem:Complexity of SharePerm}
    For the protocol $\mathrm{SharePerm}(k, P_w: \pi)$
    given in Algorithm~\ref{alg:SharePerm},
    the communication, computation, and round complexities are
    $O(\frac{k}{\log k}nm\log m)$,
    $O(\frac{k}{\log k}nm\log m)$, and
    $O(\log \log k)$, respectively.
\end{theorem}

    \begin{theorem}\label{theorem:Complexity of First Shuffle}
    Consider an input collection
    $\ll \mathbf{V} \rr = \ll \mathbf{v}_1, \ldots, \mathbf{v}_m \rr$
    of $m$ vectors, each of length~$l$, to be jointly shuffled by the parties.
    For our first shuffle protocol described in
    Algorithms~\ref{alg:FirstShuffleOff} and~\ref{alg:FirstShuffleOn},
    the offline communication, computation, and round complexities are
    $O(\frac{k}{\log k}n^2m\log m)$,
    $O(\frac{k}{\log k}n^2m\log m)$, and
    $O(\log \log k)$, respectively.
    The online communication, computation, and round complexities are
    $O(\frac{l}{\log k}n^2m\log m)$,
    $O(\frac{kl}{\log k}n^2m\log m)$, and
    $O(\frac{n\log m}{\log k})$, respectively.
    \end{theorem}

\rev{
The above bounds have a simpler interpretation in the common setting where
the number of parties $n$ and the vector length $l$ are either treated as
fixed or grow much more slowly than the number of shuffled items $m$.
The parameter $k$ specifies the size of the small permutations used in the
decomposition: an $m$-permutation is represented by
$O(\frac{m\log m}{k\log k})$ small $k$-permutations, organized into
$O(\frac{\log m}{\log k})$ decomposition layers. After suppressing factors
depending on $n$ and $l$, the offline communication and computation costs
of \textsf{Shuffle1} scale as $O(\frac{km\log m}{\log k})$, while the
online communication scales as $O(\frac{m\log m}{\log k})$, the online
computation scales as $O(\frac{km\log m}{\log k})$, and the online round
complexity scales as $O(\frac{\log m}{\log k})$.
Thus, a larger $k$ reduces the number of decomposition layers and the
number of small $k$-permutations. This improves the online communication
and online round complexity of \textsf{Shuffle1}. However, a larger $k$
also makes each $k\times k$ permutation matrix more expensive to generate
and to apply. Consequently, it increases the preprocessing cost in the
offline phase and the computation cost in both the offline and online
phases. Therefore, \textsf{Shuffle1} provides a direct trade-off: larger
$k$ reduces online communication and rounds, while smaller $k$ reduces
preprocessing and computation.
}

\section{Shuffle Protocol with Linear Online Communication} \label{sec:shuffle}
In this section, we present our second shuffle protocol.
By leveraging the maliciously secure shuffle correlation technique of
\cite{gao2024linearshuffle}, we obtain a shuffle protocol with
$O(nml)$ online communication and
$O\!\left(\frac{k+l}{\log k} \, n^2 m \log m\right)$ offline communication.
For clarity of exposition, we first describe the construction for the case
$l=1$, that is, shuffling individual vector entries.

At a high level, the semi-honest shuffle correlation operates as follows.
Suppose the parties share a vector $\ll \mathbf{x}\rr$ and an independent
random vector $\ll \mathbf{r}\rr$. Each party $P_i$ samples a random
permutation $\pi_i$ and secret-shares a random vector
$\ll \mathbf{r}_i\rr$ among all parties. In the online phase, the parties
first open
$
\ll \mathbf{y}_0\rr = \ll \mathbf{x}\rr + \ll \mathbf{r}\rr
$
to $P_1$. Subsequently, each party $P_i$ sends to $P_{i+1}$ the value
\[
\mathbf{y}_i \gets \pi_i(\mathbf{y}_{i-1}) + \mathbf{r}_i ,
\]
while the last party $P_n$ broadcasts $\mathbf{y}_n$. To recover the
shuffled output, the parties jointly remove the random mask by subtracting
\[
\ll \mathbf{\Delta}\rr
= \pi_n\bigl(\pi_{n-1}(\cdots \pi_1(\ll \mathbf{r}\rr) + \ll \mathbf{r}_1\rr \cdots)
+ \ll \mathbf{r}_{n-1}\rr\bigr) + \ll \mathbf{r}_n\rr
\]
from $\mathbf{y}_n$.

Accordingly, the primary objective of the offline phase is to generate
the correction term $\ll \mathbf{\Delta}\rr$.
\cite{gao2024linearshuffle} show that this can be accomplished efficiently
using an ideal permutation functionality
\[
\ll \pi(\mathbf{x})\rr \gets \Pi_{\mathrm{perm}}(P_i : \pi, \ll \mathbf{x}\rr),
\]
which can be instantiated by our permutation protocol in
Section~\ref{sec:permutation}. They further demonstrate how to enhance
the construction to achieve malicious security, i.e., security with abort,
by jointly shuffling $\beta \mathbf{x}$ alongside $\mathbf{x}$ and checking
the correctness of messages.

For completeness, we provide a full specification of the protocol below.
As noted above, the resulting protocols initially achieve security with abort.
In Section~\ref{sec:GOD}, we show how to further enhance them to achieve
guaranteed output delivery (GOD).

\subsection{Correlation Check of Public Values}
The goal of the protocol
$\mathrm{CorrCheck}(\mathbf{a}, \mathbf{b}, \ll \beta \rr, \ll \mathbf{r} \rr)$
is to allow a designated party $P_w$ to check whether the relation
$\beta \mathbf{a} = \mathbf{b} + \mathbf{r}$ holds.
For an honest $P_w$, the protocol aborts with overwhelming probability
whenever the correlation is violated.

This protocol is used in the online phase of the shuffle protocol, where
each party $P_i$ sends vectors $\mathbf{a}$ and $\mathbf{b}$ to $P_{i+1}$.
Party $P_{i+1}$ can invoke $\mathrm{CorrCheck}$ to ensure that the received
messages are correct. If an honest $P_{i+1}$ accepts, this guarantees that
$P_i$ has sent correct messages. As shown in \cite{gao2024linearshuffle},
as long as each honest $P_w$ receives correct messages, the permutation
chosen by $P_w$ remains hidden.

To perform the check, party $P_w$ first samples a random challenge
$\lambda \gets \mathbb{K}$ locally. It then computes
\[
u \gets \sum_{i=1}^m \lambda^{i-1}\mathbf{a}(i),
\qquad
v \gets \sum_{i=1}^m \lambda^{i-1}\mathbf{b}(i),
\]
and broadcasts $\lambda, u, v$.
All parties then jointly compute
\[
\ll s \rr \gets \ll \beta \rr \cdot u - v
    - \sum_{i=1}^m \lambda^{i-1}\ll \mathbf{r}(i)\rr,
\]
and open $\ll s \rr$ to check whether it is zero.
Apart from the broadcast of $\lambda$, the only communication comes from
calls to $\Pi_\mathrm{mul}$ and $\Pi_\mathrm{open}$.

The protocol is formally described in Algorithm~\ref{alg:CorrCheck}.

    \begin{algorithm}
        \caption{$\mathrm{CorrCheck}(P_w, \mathbf{a}, \mathbf{b}, \ll \beta \rr, \ll \mathbf{r} \rr )$}\label{alg:CorrCheck}
        \begin{algorithmic}
        \Require $\mathbf{a}$ and $\mathbf{b}$ are plaintext vectors and known to $P_w$.
        \Ensure If $P_w = \emptyset$ or $P_w$ is honest, abort w.h.p. if $\beta\mathbf{a} \neq \mathbf{b} + \mathbf{r}$.
    \State Let $m$ be the length of $\mathbf{a}$, $\mathbf{b}$, and $\ll \mathbf{r}\rr$.\;
    \If{$P_w\neq \emptyset$}
        \State $P_w$ draws $\lambda\gets \mathbb{K}$ uniformly.\;
        \State $P_w$ computes $u\gets \sum_{i=1}^{m}\lambda^{i-1}\mathbf{a}(i)$.\;
        \State $P_w$ computes $v\gets \sum_{i=1}^{m}\lambda^{i-1}\mathbf{b}(i)$.\;
        \State $P_w$ broadcasts $\lambda, u, v$.\;
    \Else \Comment{$\mathbf{a}$ and $\mathbf{b}$ are known to all parties.}
        \State $\lambda \gets \Pi_\mathrm{challenge}()$\;
        \State All parties locally compute $u\gets \sum_{i=1}^{m}\lambda^{i-1}\mathbf{a}(i)$.\;
        \State All parties locally compute $v\gets \sum_{i=1}^{m}\lambda^{i-1}\mathbf{b}(i)$.\;
    \EndIf
    \State $\ll s\rr \gets\ll \beta \rr u - v - \sum_{i=1}^m \lambda^{i-1}\ll \mathbf{r}(i)\rr$\;
    \State $\ll c \rr \gets \Pi_\mathrm{rand}(\mathbb{K})$\;
    \State $\ll s\rr \gets \Pi_\mathrm{mul}(\ll s\rr, \ll c\rr)$\;
    \State $s\gets \Pi_\mathrm{open}(\ll s \rr)$\;
    \State All parties abort if $s\neq 0$.\;
        \end{algorithmic}
    \end{algorithm}

\subsection{Offline Phase} \label{sec:offline shuffle protocol}
The main task of the offline phase is to generate a shuffle correlation,
i.e., a collection of random values satisfying specific correlation
constraints. The shuffle correlation is defined below.

\begin{definition}[$m$-Shuffle Correlation from \cite{gao2024linearshuffle}]\label{def:malicious_cor}
    The $m$-shuffle correlation is defined as
    $$ \mathrm{cor} = \left\{\begin{matrix}
        \ll \beta \rr  & \ll \mathbf{r} \rr& \ll \beta \mathbf{r} \rr & \ll \mathbf{s}\rr & \\
        \pi_1 & \pi_2 & \cdots & \pi_n\\
        \ll \pi_1(\mathbf{r}_1^\prime)\rr& \ll \pi_2(\mathbf{r}_2^\prime)\rr& \cdots & \ll \pi_n(\mathbf{r}_n^\prime)\rr \\
        &\mathbf{z}_2 & \cdots & \mathbf{z}_n\\
    \end{matrix}\right\},$$
    where
    \begin{enumerate}
        \item Each $\pi_i$ is an $m$-permutation chosen by party $P_i$ and known only to it;
        if $P_i$ is honest, then $\pi_i$ is uniform over all $m$-permutations.
        \item $\ll \beta \rr$ is a uniformly random secret-shared element of $\mathbb{K}$.
        \item $\ll \mathbf{r}\rr, \ll\mathbf{s}\rr, \ll \mathbf{r}_1^\prime\rr, ..., \ll
        \mathbf{r}_n^\prime\rr$ are secret-shared vectors of length $m$, with each
        entry independently uniformly random over $\mathbb{K}$.
        Moreover, $\ll \beta\mathbf{r}\rr$ is a secret sharing of $\beta\cdot\mathbf{r}$.
        \item $\mathbf{z}_i = (\mathbf{z}_{i,1}, \mathbf{z}_{i,2})$, where
        each $\mathbf{z}_{i,j}$ is a vector of length $m$.
        The entries of $\mathbf{z}_{i,1}$ are uniformly random, subject to the constraint
        $$\left\{\begin{aligned}
            \mathbf{z}_{i,2} &= \beta\mathbf{z}_{i,1} + \pi_{i-1}(\mathbf{r}_{i-1}^\prime) - \mathbf{r}_i^\prime\quad \forall i=2,3,...,n,\\
            \mathbf{s} &= \pi_n(\pi_{n-1}(
                \cdots \pi_2(\pi_1(\mathbf{r})-\mathbf{z}_{2,1})-
                \mathbf{z}_{3,1}\cdots)-\mathbf{z}_{n,1}).
        \end{aligned}\right.$$
    \end{enumerate}
\end{definition}

To generate a shuffle correlation, the parties first generate a random
$\ll \beta \rr$ and $2n$ random vectors of length $m$ over $\mathbb{K}$,
denoted by
$$ \ll \mathbf{r}_1 \rr, \ll \mathbf{r}_2 \rr, ..., \ll \mathbf{r}_n \rr,
\ll \mathbf{r}^\prime_1 \rr, \ll \mathbf{r}^\prime_2 \rr, ..., \ll \mathbf{r}^\prime_n \rr. $$
The parties then call $\Pi_\mathrm{mul}$ to compute the scalar-vector
product of \(\llbracket \beta \rrbracket\) with each
\(\llbracket \mathbf{r}_i \rrbracket\), yielding
$$ \ll \beta \mathbf{r}_1 \rr, \ll \beta \mathbf{r}_2 \rr, ..., \ll \beta \mathbf{r}_n \rr. $$

The parties then invoke $\mathrm{Permute}$ $n$ times, once for each
$i\in [n]$, to compute
$$
    \ll \pi_i(\mathbf{r}_i),
    \pi_i(\beta\mathbf{r}_i),
    \pi_i(\mathbf{r}^\prime_i) \rr
    \gets \mathrm{Permute}(P_i: \pi_i, \ll \mathbf{r}_i, \beta\mathbf{r}_i, \mathbf{r}^\prime_i \rr).
$$

 For each $i\geq 2$, the parties compute
$$\begin{aligned}
    \ll \mathbf{z}_{i,1} \rr &\gets \ll \pi_{i-1}(\mathbf{r}_{i-1}) \rr - \ll \mathbf{r}_i \rr,\\
    \ll \mathbf{z}_{i,2} \rr &\gets \ll \pi_{i-1}(\beta\mathbf{r}_{i-1}) \rr + \ll \pi_{i-1}(\mathbf{r}_{i-1}^\prime) \rr - \ll \beta\mathbf{r}_i \rr - \ll \mathbf{r}^\prime_i \rr.
\end{aligned}$$
Note that $\ll\mathbf{z}_i\rr\coloneqq (\ll \mathbf{z}_{i,1} \rr, \ll \mathbf{z}_{i,2} \rr)$
is a vector of length $2m$.

The parties then open each $\ll \mathbf{z}_i \rr$ to $P_i$ for each
$i\geq 2$. The shuffle correlation returned by this protocol is
$$ \mathrm{cor}\coloneqq \left\{\begin{matrix}
    \ll\beta\rr & \ll \mathbf{r}_1\rr & \ll\beta \mathbf{r}_1\rr & \ll\pi_n(\mathbf{r}_n)\rr\\
    \pi_1 & \pi_2 & ... & \pi_n\\
    \ll \pi_1(\mathbf{r}_1^\prime)\rr& \ll \pi_2(\mathbf{r}_2^\prime)\rr& \cdots & \ll \pi_n(\mathbf{r}_n^\prime)\rr \\
    & \mathbf{z}_2 & \cdots & \mathbf{z}_n\\
\end{matrix} \right\}.
$$
Note that each $\mathbf{z}_i$ is held as plaintext by party $P_i$ for
$i\geq 2$, while variables enclosed in brackets are shared among all
parties.

This protocol is formally described in Algorithm~\ref{alg:shuffle off}.

\begin{algorithm}
    \caption{$\mathrm{cor} \gets \mathrm{Shuffle2}_\mathrm{off}(P_1:\pi_1, ..., P_n:\pi_n)$}\label{alg:shuffle off}
    \begin{algorithmic}
    \Require For each honest $P_i$, $\pi_i$ is sampled uniformly from all $m$-permutations.
    \Ensure Output a shuffle correlation.
    \State $\ll \beta \rr\gets\Pi_\mathrm{rand}(\mathbb{K})$.\;
    \For{$i=1$ to $n$} \textbf{parallel}
        \State $\ll \mathbf{r}_i(j) \rr\gets \Pi_\mathrm{rand}(\mathbb{K})$ for $j\in [m]$\;
        \State $\ll \mathbf{r}_i^\prime(j) \rr\gets \Pi_\mathrm{rand}(\mathbb{K})$ for $j\in [m]$\;
        \State $\ll \beta \mathbf{r}_i \rr \gets \Pi_\mathrm{mul}(\ll \beta \rr, \ll \mathbf{r}_i\rr)$\;
        \State $\ll\pi_i\rr\gets \mathrm{SharePerm}(k, P_i:\pi_i)$\;
        \State $\ll \pi_i(\mathbf{r}_i),
            \pi_i(\beta\mathbf{r}_i),
            \pi_i(\mathbf{r}^\prime_i) \rr
            \gets \mathrm{Permute}(\ll\pi_i\rr,
            \ll \mathbf{r}_i, \beta\mathbf{r}_i, \mathbf{r}^\prime_i \rr)$\;
        \If{$i\geq 2$}
            \State $\ll \mathbf{z}_{i,1} \rr \gets \ll \pi_{i-1}(\mathbf{r}_{i-1}) \rr - \ll \mathbf{r}_i \rr$\;
            \State $\ll \mathbf{z}_{i,2} \rr \gets \ll \pi_{i-1}(\beta\mathbf{r}_{i-1}) \rr - \ll \beta\mathbf{r}_i \rr + \ll \pi_{i-1}(\mathbf{r}_{i-1}^\prime) \rr - \ll \mathbf{r}^\prime_i \rr$\;
            \State $ \ll \mathbf{z}_i \rr \coloneqq (\ll \mathbf{z}_{i,1} \rr, \ll \mathbf{z}_{i,2} \rr)$\;
        \EndIf
    \EndFor
    \For{$i=2$ to $n$} \textbf{parallel}
        \State $\mathbf{z}_i \gets \Pi_\mathrm{open}(\ll \mathbf{z}_i\rr, P_i)$\;
    \EndFor
    \State Return $\mathrm{cor}\coloneqq \left\{\begin{matrix}
        \ll\beta\rr & \ll \mathbf{r}_1\rr & \ll\beta \mathbf{r}_1\rr & \ll\pi_n(\mathbf{r}_n)\rr\\
        \pi_1 & \pi_2 & ... & \pi_n\\
        \ll \pi_1(\mathbf{r}_1^\prime)\rr& \ll \pi_2(\mathbf{r}_2^\prime)\rr& \cdots & \ll \pi_n(\mathbf{r}_n^\prime)\rr \\
        & \mathbf{z}_2 & \cdots & \mathbf{z}_n\\
    \end{matrix} \right\}$.\;
    \end{algorithmic}
\end{algorithm}

\subsection{Online Phase}
In the online phase, to shuffle $\ll \mathbf{x} \rr$, the parties first compute
$$\ll \mathbf{z}_1\rr \gets (\ll\mathbf{x}\rr-\ll\mathbf{r}_1\rr, \ll\beta\rr\ll\mathbf{x}\rr - \ll\beta\mathbf{r}_1\rr - \ll \mathbf{r}^\prime_1\rr),$$
and open $\ll \mathbf{z}_1 \rr$ to $P_1$.
Then $P_1$ applies its permutation and sends
$$ \mathbf{y}_1 \gets \pi_1(\mathbf{z}_1) $$
to party $P_2$.
Here, we slightly abuse notation by defining
$ \pi_1(\mathbf{z}_1) \coloneqq (\pi_1(\mathbf{z}_{1,1}), \pi_1(\mathbf{z}_{1,2})). $
This notation extends to all such vectors of length $2m$, although
$\pi$ is an $m$-permutation.

For $2\leq i\leq n-1$, $P_i$ first checks whether the message received
from $P_{i-1}$ is correct by invoking
$$ \mathrm{CorrCheck}(P_i, \mathbf{y}_{i-1, 1}, \mathbf{y}_{i-1, 2}, \ll\beta\rr, \ll \pi_{i-1}(\mathbf{r}^\prime_{i-1}) \rr). $$
If the check passes, $P_i$ computes
$$ \mathbf{y}_i \gets \pi_i(\mathbf{y}_{i-1} + \mathbf{z}_i) $$
and sends $\mathbf{y}_i$ to $P_{i+1}$; otherwise, all parties abort.

The last party $P_n$ performs the same check on $\mathbf{y}_{n-1}$.
If the check passes, $P_n$ broadcasts
$ \mathbf{y}_n \gets \pi_n(\mathbf{y}_{n-1} + \mathbf{z}_n) $.
Then all parties invoke
$$ \mathrm{CorrCheck}(\emptyset, \mathbf{y}_{n, 1}, \mathbf{y}_{n, 2}, \ll\beta\rr, \ll \pi_{n}(\mathbf{r}^\prime_{n}) \rr). $$
If the final check passes, the parties output the secret-shared shuffled
vector by computing
$$\ll \mathbf{x} ^\prime \rr \gets \mathbf{y}_{n,1} + \ll \pi_n(\mathbf{r}_n)\rr .$$

The above process is formally described in Algorithm~\ref{alg:shuffle on}.

To see the correctness of the protocol, note that if all parties are
honest, each $P_{i+1}$ receives $\mathbf{y}_i$ satisfying
$$ \begin{aligned}
    \mathbf{z}_{i+1, 1} &= \pi_i(\mathbf{r}_i) - \mathbf{r}_{i+1},\\
    \mathbf{z}_{i+1, 2} &= \beta \mathbf{z}_{i+1, 1} + \pi_i(\mathbf{r}^\prime_i) - \mathbf{r}^\prime_{i+1},\\
    \mathbf{y}_{i, 1} &= \pi_i\circ\cdots\circ\pi_1(\mathbf{x}) - \pi_i(\mathbf{r}_i), \\
    \mathbf{y}_{i, 2} &= \beta \mathbf{y}_{i, 1} - \pi_i(\mathbf{r}^\prime_i).
\end{aligned} $$
Then $P_{i+1}$ sends
$\mathbf{y}_{i+1} =\pi_{i+1}(\mathbf{y}_i + \mathbf{z}_{i+1})$.
The sequence $\mathbf{y}_i$ maintains the invariant
$$ \begin{aligned}
    \mathbf{y}_{i+1, 1} &= \pi_{i+1}\circ\cdots\circ\pi_1(\mathbf{x}) - \pi_{i+1}(\mathbf{r}_{i+1}), \\
    \mathbf{y}_{i+1, 2} &= \beta \mathbf{y}_{i+1, 1} - \pi_{i+1}(\mathbf{r}^\prime_{i+1}).
\end{aligned}
$$
Finally, the parties compute $\mathbf{y}_{n,1} + \pi_n(\mathbf{r}_n)$,
which is exactly $\pi(\mathbf{x})$.

\begin{algorithm}
    \caption{$\ll \mathbf{x}^\prime \rr\gets \mathrm{Shuffle2}_\mathrm{on}(\ll \mathbf{x}\rr)$}\label{alg:shuffle on}
    \begin{algorithmic}
    \Ensure Output the shuffled value $\ll\pi(\mathbf{x})\rr$ for some secret uniform $\pi$.
    \State Fetch a fresh $m$-shuffle correlation $\mathrm{cor}$ generated by $\mathrm{Shuffle2}_\mathrm{off}$.\;
    \State $\ll \mathbf{z}_1\rr \gets (\ll\mathbf{x}\rr-\ll\mathbf{r}_1\rr, \ll\beta\rr\ll\mathbf{x}\rr - \ll\beta\mathbf{r}_1\rr - \ll \mathbf{r}^\prime_1\rr)$\;
    \State Open $\ll \mathbf{z}_1\rr$ to $P_1$.\;
    \State $P_1$ computes $\mathbf{y}_1 \gets \pi_1(\mathbf{z}_1)$.\;
    \State $P_1$ sends $\mathbf{y}_1$ to $P_2$.\;
    \For{$i = 2$ to $n$}
        \State $\mathrm{CorrCheck}(P_i, \mathbf{y}_{i-1, 1}, \mathbf{y}_{i-1, 2}, \ll \beta\rr, \ll \pi_{i-1}(\mathbf{r}^\prime_{i-1})\rr)$\;
        \State $P_i$ computes $\mathbf{y}_i \gets \pi_i(\mathbf{y}_{i-1} + \mathbf{z}_i)$.\;
        \If{$i < n$}
            \State $P_i$ sends $\mathbf{y}_i$ to $P_{i+1}$.\;
        \Else
            \State $P_n$ broadcasts $\mathbf{y}_n$.\;
        \EndIf
    \EndFor
    \State $\mathrm{CorrCheck}(\emptyset, \mathbf{y}_{n, 1}, \mathbf{y}_{n, 2}, \ll \beta\rr, \ll \pi_{n}(\mathbf{r}^\prime_{n})\rr)$\;
    \State Return $\mathbf{y}_{n, 1} + \ll \pi_n(\mathbf{r}_n)\rr$.\;
    \end{algorithmic}
\end{algorithm}

\paragraph{Security of \textsf{Shuffle2}.}
\rev{
The shuffle correlation moves most of the expensive permutation work to the
offline phase while preserving the same security guarantee as a direct
permute-in-turn shuffle. The online messages are masked by independent
random values prepared in the correlation, so revealing or forwarding them
does not expose the underlying input vector or the honest parties'
permutations. The correctness of the online chain is enforced by the hidden
scalar \(\beta\) and the auxiliary masks \(\mathbf r'_i\). If a corrupted
party sends an inconsistent message, then the relation checked by
\textsf{CorrCheck} becomes false. Since \(\beta\) is unknown to the
adversary and the check is compressed by a fresh random challenge, such an
inconsistency is rejected with overwhelming probability. Consequently,
conditioned on all checks accepting, the online chain is consistent with the
preprocessed shuffle correlation and applies exactly the composed
permutation to the input, while the permutation remains hidden. This
argument follows the shuffle-correlation security of
\cite{gao2024linearshuffle}, together with the security of our
\textsf{Permute} protocol. The detailed security proofs can be found in
Appendix~\ref{sec:security proofs}.
}

\subsection{Complexity Analysis}\label{subsec:complexity of sec5}
We analyze the complexity of the above shuffle protocol, distinguishing
explicitly between its offline and online phases. The main results are
stated below, and the proofs are provided in
Appendix~\ref{subsec:proof of sec5}.
\begin{theorem}
    For $\mathrm{Shuffle2}_\mathrm{off}(P_1:\pi_1, ..., P_n:\pi_n)$
    described in Algorithm~\ref{alg:shuffle off},
    the offline communication, computation, and round complexities are
    $O(\frac{k}{\log k}n^2m\log m)$,
    $O(\frac{k}{\log k}n^2m\log m)$, and
    $O(\frac{\log m}{\log k}+\log\log k)$, respectively.
\end{theorem}

\begin{theorem}
    For $\mathrm{Shuffle2}_\mathrm{on}(\ll \mathbf{v}\rr)$ described in
    Algorithm~\ref{alg:shuffle on}, the communication, computation, and
    round complexities are $O(nm+n^2)$, $O(n^2m)$, and $O(n)$,
    respectively. Note that when $n=O(m)$, the communication complexity
    simplifies to $O(nm)$.
\end{theorem}

\begin{corollary}
    For shuffling the rows of an $m\times l$ matrix, the shuffle protocol
    described in Algorithms~\ref{alg:shuffle off} and~\ref{alg:shuffle on}
    requires $O(nml+n^2)$ online communication,
    $O(n^2ml)$ online computation, and $O(n)$ online rounds,
    with an offline overhead of
    $O(\frac{k+l}{\log k}n^2m\log m)$ communication,
    $O(\frac{kl}{\log k}n^2m\log m)$ computation, and
    $O(\frac{\log m}{\log k}+\log\log k)$ rounds.
\end{corollary}
\rev{
The above bounds can be interpreted more directly in common parameter
settings. For the online phase of \textsf{Shuffle2}, the communication complexity
$O(nml+n^2)$ simplifies to $O(nml)$ in typical MPC deployments where $n$
is much smaller than $ml$, since the $n^2$ term is then dominated by $nml$.
In particular, if $n$ and $l$ are either treated as fixed or grow
much more slowly than $m$, the online communication and online computation
both scale linearly with the number of shuffled items $m$, while the online
round complexity is independent of $m$.
For the offline phase, after suppressing factors depending only on the
number of parties $n$, the offline communication scales as
$O(\frac{(k+l)m\log m}{\log k})$, the offline computation scales as
$O(\frac{klm\log m}{\log k})$, and the offline round complexity scales as
$O(\frac{\log m}{\log k}+\log\log k)$.
}

\paragraph{Choosing the parameter $k$.}
\rev{
For \textsf{Shuffle1}, a larger $k$ reduces online communication and rounds,
whereas a smaller $k$ reduces preprocessing and computation. For
\textsf{Shuffle2}, $k$ does not affect the online phase and only tunes the
offline phase. For both protocols, a simple empirical cost model can guide
the choice of $k$: the $k$-dependent part of the total running time can be
approximated as proportional to
$\frac{(1+\frac{1}{3l})k+r}{\log k}$, where $r$ is a platform-dependent
coefficient reflecting the ratio between local computation throughput and
network bandwidth.
}

\rev{
In practice, $k$ can be tuned
empirically according to the performance objective and system bottleneck.
A rough estimate of $r$ can be obtained as the ratio between the single-party
integer-operation throughput, measured in MOps/sec, and the pairwise network
bandwidth, measured in MB/sec. One may then test powers of 2 near the
minimizer of $\frac{(1+\frac{1}{3l})k+r}{\log k}$ and choose the best value
under the target deployment. Since $r$ is often large in practice, the total
running time is usually not very sensitive to $l$ when $l$ is large, nor to
moderate deviations from the optimal value of $k$. When the target $n$, $m$,
or $l$ is large, these trials can be performed on smaller representative
instances to reduce tuning time.
}

\subsection{Evaluation}\label{subsec:Evaluation}

The theoretical analysis presented above already indicates
that our scheme achieves excellent performance.
In particular, it achieves state-of-the-art asymptotic complexity,
while maintaining small constant factors, suggesting strong practical efficiency.
To provide a more concrete and intuitive demonstration of these advantages,
we present a simple experimental evaluation below.

We compare our two protocols, \textsf{Shuffle1} and \textsf{Shuffle2},
with the protocols of~\cite{lu2023rpm} and~\cite{keller2014ORAM}.
The construction in~\cite{lu2023rpm} produces permutations that are statistically close to uniform
and is specifically optimized for online performance,
while the protocol of~\cite{keller2014ORAM} achieves exact uniformity and represents the previous state of the art.
The source code of our experiments is publicly available\footnote{\url{https://github.com/GJCPP/ShamirShuffle}}.

\begin{landscape}
\begin{table}
\centering
\caption{Performance metrics vs. number of vectors ($m$) with $l=1$.
``---'' denotes missing data due to timeout (1200\,s).
For the same $m$, the best values for the online-related metrics and total
time are bolded.}
\label{tab:m_variation}
\scriptsize
\setlength{\tabcolsep}{2.0pt}
\begin{tabular}{|c|c|c|c|c|c|c|c|c|}
\hline
$m$ & Protocol & Online & Offline & Online & Offline & Online & Offline & Total \\
& & Comm (MB) & Comm (MB) & Rounds & Rounds & Time (s) & Time (s) & Time (s) \\

\hline
16 & \texttt{\cite{lu2023rpm}} & 0.004 & 1.213 & 32 & 22 & 0.643 & 0.624 & 1.267 \\
& \texttt{\cite{keller2014ORAM}} & 0.013 & 0.055 & 70 & 8 & 1.404 & 0.172 & 1.576 \\
& \texttt{Shuffle1} & 0.004 & 0.035 & 20 & 3 & 0.402 & 0.072 & 0.474 \\
& \texttt{Shuffle2} & \textbf{$<$0.001} & 0.042 & \textbf{6} & 8 & \textbf{0.120} & 0.178 & \textbf{0.298} \\
\hline
256 & \texttt{\cite{lu2023rpm}} & 0.098 & 70.164 & 32 & 22 & 0.675 & 6.406 & 7.081 \\
& \texttt{\cite{keller2014ORAM}} & 0.468 & 1.875 & 150 & 16 & 3.031 & 0.560 & 3.591 \\
& \texttt{Shuffle1} & 0.094 & 1.406 & 30 & 4 & 0.624 & 0.299 & 0.923 \\
& \texttt{Shuffle2} & \textbf{0.004} & 1.555 & \textbf{6} & 11 & \textbf{0.121} & 0.498 & \textbf{0.618} \\
\hline
1024 & \texttt{\cite{lu2023rpm}} & 0.391 & 550.863 & 32 & 22 & 0.804 & 50.691 & 51.496 \\
& \texttt{\cite{keller2014ORAM}} & 2.372 & 9.500 & 190 & 20 & 3.954 & 1.285 & 5.238 \\
& \texttt{Shuffle1} & 0.624 & 5.875 & 50 & 6 & 1.091 & 0.945 & 2.035 \\
& \texttt{Shuffle2} & \textbf{0.016} & 6.718 & \textbf{6} & 17 & \textbf{0.122} & 1.291 & \textbf{1.414} \\
\hline
4096 & \texttt{\cite{lu2023rpm}} & 1.563 & 4367.163 & 32 & 22 & 1.681 & 486.224 & 487.905 \\
& \texttt{\cite{keller2014ORAM}} & 11.499 & 46.000 & 230 & 24 & 5.316 & 4.775 & 10.092 \\
& \texttt{Shuffle1} & 4.000 & 27.000 & 80 & 9 & 2.024 & 3.683 & 5.707 \\
& \texttt{Shuffle2} & \textbf{0.063} & 31.875 & \textbf{6} & 26 & \textbf{0.129} & 4.492 & \textbf{4.621} \\
\hline
16384 & \texttt{\cite{lu2023rpm}} & --- & --- & --- & --- & --- & --- & ---\\
& \texttt{\cite{keller2014ORAM}} & 53.995 & 216.000 & 270 & 28 & 8.626 & 20.817 & 29.443 \\
& \texttt{Shuffle1} & 17.998 & 122.000 & 90 & 10 & 3.594 & 15.778 & 19.373 \\
& \texttt{Shuffle2} & \textbf{0.250} & 143.498 & \textbf{6} & 29 & \textbf{0.158} & 18.267 & \textbf{18.425} \\
\hline
65536 & \texttt{\cite{lu2023rpm}} & --- & --- & --- & --- & --- & --- & ---\\
& \texttt{\cite{keller2014ORAM}} & 247.999 & 992.000 & 310 & 32 & 20.463 & 92.394 & 112.858 \\
& \texttt{Shuffle1} & 127.999 & 632.000 & 160 & 17 & 12.385 & 69.447 & 81.832 \\
& \texttt{Shuffle2} & \textbf{1.000} & 773.999 & \textbf{6} & 50 & \textbf{0.274} & 80.083 & \textbf{80.357} \\
\hline
262144 & \texttt{\cite{lu2023rpm}} & --- & --- & --- & --- & --- & --- & ---\\
& \texttt{\cite{keller2014ORAM}} & 1119.994 & 4480.000 & 350 & 36 & 71.133 & 425.162 & 496.295 \\
& \texttt{Shuffle1} & 575.997 & 2848.000 & 180 & 19 & 45.586 & 315.835 & \textbf{361.421} \\
& \texttt{Shuffle2} & \textbf{4.000} & 3479.997 & \textbf{6} & 56 & \textbf{0.761} & 363.395 & 364.156 \\
\hline
\end{tabular}
\end{table}
\end{landscape}

\rev{
All protocols are implemented in \texttt{C++}, based on Shamir secret
sharing, in the semi-honest security model. The benchmark implementation
captures the core execution cost of the proposed shuffle constructions,
including permutation sharing, permutation application, and shuffle
correlation. This level of implementation is sufficient to illustrate
both the asymptotic and practical advantages of our constructions in the
dominant protocol operations, especially online communication, online
round complexity, and concrete online latency. It does not include the
additional checking, tag-generation, and state-maintenance mechanisms
needed for malicious security and guaranteed output delivery (GOD,
introduced in Section~\ref{sec:GOD}); these mechanisms are expected to
result only in modest constant-factor overheads. A more detailed
discussion is provided in Appendix~\ref{app:implementation-scope}.
}

All experiments are conducted over a 128-bit prime field with statistical
security parameter $\kappa = 40$.
Experiments are run on a host equipped with 32 processors, each an
Intel(R) Xeon(R) Gold 5122 CPU at 3.60\,GHz, running Ubuntu 18.04.6 LTS.
Each party is simulated as an independent process, and communication is
implemented via the Message Passing Interface (MPI). To emulate a
wide-area network (WAN) setting, each pair of parties is connected with a
bandwidth of 50\,MB/s and a round-trip time (RTT) of 20\,ms.

Table~\ref{tab:m_variation} presents the performance metrics as the
number of vectors ($m$) varies from 16 to 262,144, while keeping the
vector length ($l$) fixed at 1. The results demonstrate the significant
advantages of our protocols, particularly \textsf{Shuffle2}.

The key observations from Table~\ref{tab:m_variation} are:

\begin{itemize}
\item \textbf{Online Communication}: \textsf{Shuffle2} achieves a dramatic
reduction in online communication compared to prior work. For
$m = 262{,}144$, it requires only 4~MB, which is $280\times$ smaller
than~\cite{keller2014ORAM} (1120~MB). Compared to~\cite{lu2023rpm},
which is also specifically optimized for the online phase, \textsf{Shuffle2}
still achieves significantly lower online communication, demonstrating a
substantial improvement in online efficiency. In addition, \textsf{Shuffle2}
reduces online communication by $144\times$ compared to \textsf{Shuffle1}
(576~MB), showing that our refinement yields a substantial improvement
even over a strong base construction.

    \item \textbf{Online Rounds}: \textsf{Shuffle2} uses a constant number of online rounds (6), independent of $m$.
    In contrast, the number of rounds in \textsf{Shuffle1} and~\cite{keller2014ORAM} increases with the dataset size,
    reaching up to 180 and 350 rounds, respectively, at $m = 262{,}144$,
    while~\cite{lu2023rpm} also incurs a higher round complexity (e.g., 32 rounds).
    This constant-round property makes \textsf{Shuffle2} particularly suitable for high-latency environments.

\item \textbf{Online Time}: As a direct consequence of reduced communication
and constant rounds, \textsf{Shuffle2} achieves significantly lower online
latency. Even for $m = 262{,}144$, the online execution time remains below
1~second, whereas other protocols require tens to hundreds of seconds.
Although~\cite{lu2023rpm} reduces some online overhead, its online latency
remains noticeably higher than that of \textsf{Shuffle2}.

\item \textbf{Overall Cost}:
From an overall perspective, \textsf{Shuffle1} already achieves a significant
improvement over prior work such as~\cite{keller2014ORAM} in both total cost
and online efficiency, establishing a strong and practically efficient
foundation for secure shuffling. Building upon this foundation,
\textsf{Shuffle2} further optimizes the online phase while largely preserving
the total overhead. In contrast, although~\cite{lu2023rpm} improves the
online phase to some extent, it incurs prohibitively large offline
communication and runtime, which makes it impractical for large-scale
settings, as reflected by timeouts for large $m$. Therefore,
\textsf{Shuffle2} provides a strict improvement in online efficiency while
preserving a competitive overall cost. Generally, \textsf{Shuffle2} achieves
an excellent balance between online efficiency and total cost.
\end{itemize}

In summary, \textsf{Shuffle2} substantially improves the online phase while
maintaining a competitive total cost. It achieves low online communication,
constant online rounds, and sub-second online latency even for large
instances. These properties make it particularly suitable for latency-sensitive
secure computation and high-latency network settings.

\paragraph{When is \textsf{Shuffle1} preferable?}
\rev{
Although \textsf{Shuffle2} is preferable when online efficiency is the main
objective, \textsf{Shuffle1} remains a practical choice in some scenarios,
for example, when the online-latency requirement is relaxed and the overall
cost is prioritized. It directly applies our permutation-sharing primitive in the
permute-in-turn paradigm and avoids the shuffle correlation machinery,
thereby giving a simpler construction with less additional preprocessing
overhead. As a result, \textsf{Shuffle1} can have smaller constant factors
in its concrete overall communication and computation costs. This advantage
becomes more relevant in malicious or GOD instantiations, where
\textsf{Shuffle1} incurs less additional overhead; see
Appendix~\ref{app:implementation-scope}.
}

\section{Guaranteed Output Delivery}\label{sec:GOD}
This section focuses on enhancing the aforementioned secure-with-abort protocols
to achieve guaranteed output delivery (GOD),
a critical security property
that ensures honest parties obtain correct outputs
regardless of corrupted parties' behavior.
Our high-level approach is as follows.
For the protocols constructed purely from standard MPC primitives,
we can achieve GOD directly by replacing secure-with-abort
components with their GOD-enabled counterparts (e.g., \cite{goyal2020guaranteed}).
For the more specialized scenario --- namely,
the online phase of the second shuffle protocol,
where parties perform local computations and
exchange intermediate results ---
we further leverage the dispute control technique proposed
by \cite{beerliova2006dispute} to achieve GOD.
Consequently, all protocols proposed in this work are enhanced
from secure-with-abort to GOD protocols.

We next present the two key ingredients of our construction.
First, our scheme fully supports corruption identification,
which in turn enables GOD.
Second, our scheme leverages the authentication tag technique
of \cite{Ben-Sasson2012miniMPC},
thereby enabling the use of circuit division to achieve GOD
with an overhead linear in the circuit size.
In particular, authentication tags serve as commitments
to the shares produced within each circuit segment,
allowing honest parties to recover
a correct intermediate state in the presence of adversarial deviations,
rather than restarting.
Consequently, our GOD-enabled \textsf{SLIDE} achieves linear online
communication and rounds,
matching those of the linear-online-communication shuffle proposed in Section~\ref{sec:shuffle}.
Additionally,
we introduce further optimizations for our GOD-enabled protocols.

\subsection{Corruption Identification}
The key to achieving GOD lies in identifying corrupted parties
and subsequently excluding them from further computation.
We now show that our construction is robust enough to support
    corruption identification,
    in the sense that any failure of a correctness
    check will expose either a new corrupted party or a new
disputing pair.

The first step is to equip all underlying primitives with corruption
identification.
Concretely, each primitive supported by the ideal functionality
$\mathcal{F}_{\mathrm{Shamir}}$ is enhanced so that, upon failure, it reveals
either a new corrupted party or a new disputing pair, rather than
simply aborting.
For example, the multiplication protocol
$\Pi_{\mathrm{mul}}(\ll a\rr, \ll b\rr)$ now terminates in one of the following
three outcomes:
\begin{enumerate}
    \item The multiplication succeeds, and all honest parties hold
        consistent shares of $\ll ab\rr$.
    \item The multiplication fails and publicly identifies a new corrupted
        party $P_i \notin \mathrm{Corr}$.
    \item The multiplication fails and publicly identifies a new disputing
        pair $\{P_i, P_j\} \notin \mathrm{Disp}$.
\end{enumerate}
It is crucial that any reported corrupted party or disputing pair must be new,
so that the size of either $\mathrm{Corr}$ or $\mathrm{Disp}$ strictly
increases.
Moreover, these primitives are sound: an honest party is never falsely
identified as corrupted, nor will two honest parties ever be put into
dispute.
For primitives with public input, deviating from the publicly agreed input
is also detected as a failure.
Such primitives are provided by \cite{goyal2020guaranteed}.

We now argue that all protocols developed so far can be extended
to support corruption identification by replacing ``abort'' with identification of corrupted parties
or disputing pairs.

Recall that the generation of one-hot vectors relies on
$\mathrm{BinCheck}$ (Algorithm~\ref{alg:BinCheck}),
$\mathrm{Demux}$ (Algorithm~\ref{alg:Demux}), and
$\mathrm{OneHot}$ (Algorithm~\ref{alg:OneHot}).
These protocols only invoke basic MPC primitives and recursively call each
other; therefore, they directly inherit corruption identification.

For permutation sharing, we developed
$\mathrm{PermMat}$ (Algorithm~\ref{alg:PermMat}),
$\mathrm{PermCheck}$ (Algorithm~\ref{alg:PermCheck}),
$\mathrm{SharePerm}$ (Algorithm~\ref{alg:SharePerm}), and
$\mathrm{Permute}$ (Algorithm~\ref{alg:Permute}).
Protocols $\mathrm{PermMat}$ and $\mathrm{Permute}$ consist solely of calls
to previously established primitives.
If $\mathrm{PermCheck}$ fails with a nonzero final check value (as
opposed to a failure in a subroutine), then the party $P_w$ who specified
the permutation matrix must be corrupted.
Similarly, although the decomposition in $\mathrm{SharePerm}$ is performed
locally by $P_w$, any malformed permutation will be detected during
$\mathrm{PermMat}$ and $\mathrm{PermCheck}$, thereby identifying $P_w$ as
corrupted.
Conversely, if all checks pass, then the reconstructed permutation $\pi$ is
well-formed and serves as a valid input for $P_w$.

Building on these enhanced protocols, the first shuffle protocol
introduced in Section~\ref{sec:permutation}
therefore supports corruption identification.

For the shuffle protocol with linear communication complexity, we further
developed $\mathrm{CorrCheck}$ (Algorithm~\ref{alg:CorrCheck}),
$\mathrm{Shuffle2}_{\mathrm{off}}$ (Algorithm~\ref{alg:shuffle off}), and
$\mathrm{Shuffle2}_{\mathrm{on}}$ (Algorithm~\ref{alg:shuffle on}).
If $\mathrm{CorrCheck}$ fails due to a primitive failure, corrupted parties
are identified immediately.
The offline protocol again consists only of calls to primitives that
already support corruption identification.

When $\mathrm{CorrCheck}$ fails because a nonzero value $s \neq 0$ is
revealed, the misbehaving party is not immediately determined.
However, during $\mathrm{Shuffle2}_{\mathrm{on}}$, if party $P_i$ observes
such a failure, it immediately enters a dispute with $P_{i-1}$.
Indeed, as shown in \cite{gao2024linearshuffle}, if any party $P_j$
misbehaves by sending an incorrect message $\mathbf{y}_j$ to an honest
$P_{j+1}$, the correlation check at $P_{j+1}$ fails with overwhelming
probability.
Since all previous checks have passed, $P_{i-1}$ must have provided correct
inputs to $\mathrm{CorrCheck}$.
Thus, a current failure implies either \(P_{i-1}\) sent an
incorrect \(\mathbf{y}_{i-1}\) to \(P_i\), or \(P_i\)
used an incorrect \(\mathbf{y}_{i-1}\)
in the correlation check; hence, at least one of
$P_{i-1}$ or $P_i$ is corrupted.

If $P_{i-1}$ and $P_i$ are in dispute, the value $\mathbf{y}_{i-1}$ must
be relayed through a third party $P_{(i-1)\leftrightarrow i}$ that does not dispute
with either.
Letting \(P_{(i-1)\leftrightarrow i}\) access \(\mathbf{y}_{i-1}\)
does not compromise security,
since the analysis in \cite{gao2024linearshuffle} shows
that security is preserved even if all $\mathbf{y}_i$ values are broadcast
(broadcasting is avoided only to reduce communication overhead).
If the correlation check for $P_i$ fails again, the relay must
enter a public dispute with either
$P_{i-1}$ or $P_i$, thereby increasing $\mathrm{Disp}$.

Finally, if the last correlation check fails, then $P_n$ must be corrupted,
as all prior checks have succeeded and $P_n$ must have received correct
inputs.

In summary, all protocols proposed in this paper support corruption
identification.
Consequently, they can be strengthened to achieve GOD.

\subsection{Authentication Tag}
It remains to show that the authentication tag technique
of \cite{Ben-Sasson2012miniMPC}, presented below,
is applicable to our constructions and enables GOD with worst-case
computation linear in the circuit size.

\paragraph{Authentication Tag.}
To enable party $P_w$ to authenticate its share of a secret $\ll x\rr$,
the parties first decompose the secret as
\[
\ll x\rr = \sum_{i=1}^n \ll x(i)\rr,
\]
where $\ll x(i)\rr$ was shared earlier by party $P_i$.
Such a decomposition is always possible since all secrets in our
protocols are obtained via linear combinations of previously shared values.

For each $\ll \mathbf{x}(i)\rr$, party $P_w$ obtains one authentication
tag from each party $P_v$ that is not in dispute with $P_w$, where such a
party $P_v$ is referred to as a verifier.
Specifically, each verifier $P_v$ samples random values $\pmb{\mu}_{v\to w}$ and $u_{v\to w}$ from $\mathbb{K}$
and assists $P_w$ in computing
\[
\tau_{v\to w} \coloneqq \pmb{\mu}_{v\to w} \odot \mathbf{x}_w(i) + u_{v\to w},
\]
where $\mathbf{x}_w(i)$ denotes $P_w$'s share of $\mathbf{x}(i)$ and
$\odot$ denotes the inner product.
This computation is carried out using a mini-MPC protocol~\cite{Ben-Sasson2012miniMPC}, which is compatible with Shamir secret sharing.

Without knowledge of $(\pmb{\mu}_{v\to w}, u_{v\to w})$,
$P_w$ can forge a valid authentication tag only with probability at most $2^{-\kappa}$.
Consequently, if an inconsistency in $\ll \mathbf{x}\rr$ is later detected
and traced to $\ll \mathbf{x}(i)\rr$, party $P_i$ must accuse some party $P_j$
of misbehavior.
Upon an accusation, $P_j$ reveals $\mathbf{x}_j(i)$ together with the
corresponding authentication tags to all verifiers.
Each verifier independently checks the tag and votes to identify either $P_i$ or $P_j$ as corrupted;
a majority vote identifies a corrupted party.
Revealing $\mathbf{x}_j(i)$ is safe, since at least one of the two parties is corrupted.

\vspace{0.5em}

This technique is applicable to all our protocols, as they satisfy the
following property: upon termination, all honest parties hold
consistent shares.
Indeed, all shared values originate either from $\Pi_{\mathrm{share}}$
or from linear combinations thereof.
Any misbehavior occurring after protocol completion is therefore
handled by the authentication tag mechanism, which suffices to achieve GOD
with linear overhead.

\subsection{Optimizations}
Let $\mathcal{P}_\mathrm{active} = [n]\setminus \mathrm{Corr}$ denote the set of
parties remaining in the computation.
As observed in \cite{goyal2020guaranteed}, communication can be further reduced
by randomly selecting a subset $\mathcal{T}\subseteq \mathcal{P}_\mathrm{active}$
of size $t+1$ and letting only parties in $\mathcal{T}$ perform most of the computation.
This is sound since $t+1$ shares suffice to reconstruct a secret, and at least
one of these parties is honest.
The same optimization applies to our shuffle protocols.
Instead of having all parties permute in turn, under an honest-majority assumption,
it suffices to let only $t+1$ parties apply permutations sequentially to ensure
the uniformity of the resulting shuffle.
Consequently, the shuffle correlation can be generated for a smaller set of parties,
leading to reduced communication and computation costs.

We further remark that the offline and online phases of a shuffle protocol
can be cleanly separated into multiple segments.
This is because all secret-shared values in the correlation remain information-theoretically
hidden even if $\mathrm{CorrCheck}$ fails:
the opened value $s$ is either zero or a uniformly random element,
and thus reveals no information.
As a result, a shuffle correlation can be safely reused after a failure.

Finally, even protocols that appear atomic at a high level
(e.g., $\mathrm{Shuffle2}_\mathrm{off}$) admit decomposition into smaller
computation segments. This is because they ultimately consist of calls to
basic MPC primitives and therefore induce a divisible circuit structure.
They can thus be partitioned without compromising correctness or integrity,
which provides additional flexibility in implementation.

\section{Conclusion}

In this paper, we propose \textsf{SLIDE}, the first shuffle protocol that
simultaneously achieves uniformity, $O(nml)$ online communication, and
guaranteed output delivery (GOD). The core of \textsf{SLIDE} is our novel
permutation sharing technique, with which we construct a shuffle protocol
that outperforms all existing uniform shuffle constructions for Shamir
secret sharing. We further integrate the shuffle correlation technique of
\cite{gao2024linearshuffle} to obtain linear online complexity, and
finally extend the construction to GOD, completing the \textsf{SLIDE}
framework. As the shuffle operation is a fundamental building block for a
broad class of MPC protocols and primitives, our results have wide
applicability and directly benefit many practical MPC frameworks.

\section*{Acknowledgments}
This work was supported by the Natural Science Foundation on Frontier Leading Technology Basic Research Project of Jiangsu (No. BK20222001), the NSFC-62272215 and the 111 Center (No. B26023). Yuan Zhang and Sheng Zhong are the corresponding authors.

\bibliographystyle{SageH}
\bibliography{bib}
\appendix

\section{Security Proofs}\label{sec:security proofs}

\newcommand{\llb}{\llbracket}
\newcommand{\rrb}{\rrbracket}
\newcommand{\perm}{\Pi_\mathrm{perm}}
\newcommand{\open}{\Pi_\mathrm{open}}
\newcommand{\inner}{\Pi_\mathrm{inner}}
\renewcommand{\add}{\Pi_\mathrm{add}}
\newcommand{\challenge}{\Pi_\mathrm{challenge}}
\newcommand{\mul}{\Pi_\mathrm{mul}}
\newcommand{\Fshamir}{\mathcal{F}_\mathrm{Shamir}}

In this section, we analyze our protocols and prove that the proposed
shuffle protocols are secure. By virtue of the ideal functionality
$\Fshamir$ and the result of \cite{gao2024linearshuffle}, the security
largely follows from the definitions. Nevertheless, we provide some
intuition for the security of our protocols.

To simplify the analysis, we focus on detecting the adversary's
misbehavior. That is, we show that if the adversary deviates from the
protocol, the protocol aborts with overwhelming probability. This approach
is justified by standard simulation-based security: one constructs a
simulator that can simulate the adversary's view without knowing the
honest parties' inputs. In Shamir secret sharing, this simulation is
typically straightforward, because the simulator represents all honest
parties and can therefore reconstruct all values shared by the adversary.
Thus, in a typical proof, the simulator sends random values to the
adversary as the messages it receives, aborts if any check fails, and
finally modifies the honest parties' shares to make them consistent with
the output of the ideal functionality. It therefore suffices to argue that
any misbehavior leads to an abort with overwhelming probability, ensuring
that the simulator's output distribution coincides with the real
distribution.

Note that all basic operations in our construction rely on ideal
functionalities, which provide very strong security guarantees; for
example, even introducing an error in $\add$ triggers an immediate abort.
This strong condition is justified by the composition theorem, which
states that a protocol remains secure after ideal functionalities are
replaced with their secure implementations. Thus, any misbehavior against
the ideal functionalities is detected, and it remains only to prove that
misbehavior outside the ideal functionalities is detected.

Accordingly, we show that the protocol ends in one of two cases:
(1) the adversary misbehaves, and the protocol aborts with overwhelming
probability; or
(2) all parties act honestly, and the honest parties receive correct shares.

As shown by \cite{gao2024linearshuffle}, if the $\mathrm{Permute}$
protocol is secure, then $\mathrm{Shuffle2}_\mathrm{off}$ and
$\mathrm{Shuffle2}_\mathrm{on}$ are also secure. This is formalized in
Theorem~\ref{th:security of shuffle} below.

\begin{theorem}[\cite{gao2024linearshuffle}, informal]\label{th:security of shuffle}
Consider the ideal functionality
$$\llb \pi(\mathbf{v}) \rrb \gets \perm(P_w: \pi, \llb \mathbf{v} \rrb),$$
which takes a secret input $\pi$ from $P_w$ and securely permutes
$\llb \mathbf{v} \rrb$ into $\llb \pi(\mathbf{v}) \rrb$.

If $\mathrm{Permute}$ securely implements $\perm$, then
$\mathrm{Shuffle2}_\mathrm{off}$ and $\mathrm{Shuffle2}_\mathrm{on}$ are
secure by the composition theorem. More precisely:
\begin{itemize}
    \item If corrupted parties misbehave, all parties abort with
    overwhelming probability.
    \item If the protocol does not abort, then with overwhelming
    probability, $\mathbf{v}$ is correctly shuffled and shared among the
    honest parties, and the adversary learns no information about the
    applied permutation.
\end{itemize}
\end{theorem}

This theorem allows us to focus on the security of the $\mathrm{Permute}$
protocol. Since $\mathrm{Permute}$ only invokes basic primitives and
previously developed protocols, it suffices to prove the security of its
subprotocols.

\subsection{Security of the One-Hot Protocol}
The security of the one-hot vector sharing protocol follows almost
directly from the basic primitives, which we briefly review below.

\begin{lemma}\label{le:Security of BinCheck}
The protocol $\mathrm{BinCheck}(\llb b_1 \rrb, \dots, \llb b_d \rrb)$
described in Algorithm~\ref{alg:BinCheck} is secure in the following
sense: if any input is non-binary (i.e., not in $\{0, 1\}$), the protocol
aborts with probability
$$ p \geq 1 - \frac{d-1}{2^\kappa}. $$
\end{lemma}
\begin{proof}
Since the computation uses the ideal functionalities $\open$ and $\inner$,
the parties correctly receive the opened value
$$ u = \sum_{i=1}^d b_i \cdot (1-b_i) \cdot \lambda^{i-1}. $$

The challenge $\lambda$ is uniformly random over $\mathbb{K}$, since it
is generated by $\challenge$. In $\mathbb{K}$, the product $b_i(1-b_i)$
is zero if and only if $b_i \in \{0, 1\}$. If any $b_i$ is non-binary,
then the polynomial
$$F(X) = \sum_{i=1}^d b_i(1-b_i) X^{i-1}$$
is nonzero and has degree at most $d-1$. A nonzero polynomial of degree
$d-1$ has at most $d-1$ roots in $\mathbb{K}$, so the probability that
$\lambda$ is a root is
$$ p_\mathrm{fail} \leq \frac{d-1}{|\mathbb{K}|} \leq \frac{d-1}{2^\kappa}. $$
Hence, $u \neq 0$, and the protocol aborts with probability
$$ p = 1 - p_\mathrm{fail} \geq 1 - \frac{d-1}{2^\kappa}. $$
\end{proof}

\begin{lemma}\label{le:Security of Demux}
If all inputs are shared bits, then the protocol
$\mathrm{Demux}(\llb b_1 \rrb, \dots, \llb b_d \rrb)$ described in
Algorithm~\ref{alg:Demux} is secure in the following sense: if the
protocol does not abort, the honest parties share
$\llb \mathbf{s} \rrb$ as a one-hot vector with index
$1 + \sum_{i=1}^d 2^{i-1}b_i$.
\end{lemma}
\begin{proof}
For the base case $d=1$, the honest parties hold
$$ \llb \mathbf{s} \rrb = (\llb 1 - b_1 \rrb, \llb b_1 \rrb), $$
which is correct and secure.

For $d>1$, the protocol essentially computes the tensor product of two
one-hot vectors via calls to the ideal functionality $\mul$. Correctness
and security follow directly.
\end{proof}

\begin{theorem}\label{th:Security of OneHot}
The protocol $\mathrm{OneHot}(k, P_w: \mathrm{ind})$ described in
Algorithm~\ref{alg:OneHot} is secure in the following sense: upon
successful termination, the honest parties hold a valid one-hot vector
with the index chosen by $P_w$.

Additionally, if $P_w$ is honest, then corrupted parties learn no
information about the chosen index.
\end{theorem}
\begin{proof}
If $P_w$ is honest, then upon successful termination, the result is a
one-hot vector with the index chosen by $P_w$. Moreover,
$\mathrm{BinCheck}$ opens only a zero value, which reveals only that the
inputs are binary.

If $P_w$ is malicious and shares $\llb b_i' \rrb$ during the
$\mathrm{OneHot}$ protocol, there are two cases:
\begin{itemize}
    \item If any $b_i' \notin \{0, 1\}$, then $\mathrm{BinCheck}$ aborts
    with overwhelming probability by Lemma~\ref{le:Security of BinCheck}.
    \item If all $b_i' \in \{0, 1\}$, then there exists a valid index
    $\mathrm{ind}' = \sum_{i=1}^d 2^{i-1}b_i' + 1$ for the generated
    vector. Since $P_w$ may freely choose the index, this does not
    compromise security.
\end{itemize}
\end{proof}

\subsection{Security of the Permutation Protocol}
\begin{theorem}\label{th:Security of Permutation Matrix Check}
If the input $k$ one-hot vectors do not form a valid $k \times k$
permutation matrix, then the protocol
$\mathrm{PermCheck}(\llb \mathbf{P} \rrb)$ described in
Algorithm~\ref{alg:PermCheck} aborts with probability
$$ p \geq 1 - \frac{k-1}{2^\kappa}. $$
\end{theorem}
\begin{proof}
The parties compute
$$ \llb \mathrm{sum} \rrb \gets \sum_{j=1}^k \lambda^{j-1} \left( 1 - \sum_{i=1}^k \llb \mathbf{P}(i, j) \rrb \right). $$
Since these arithmetic operations use ideal functionalities, the parties
receive the correct $\llb \mathrm{sum} \rrb$, unless the protocol has
already aborted due to misbehavior.

\begin{itemize}
    \item If $\mathbf{P}$ is a valid permutation matrix, then
    $\sum_{i=1}^k \mathbf{P}(i, j) = 1$ for all $j$, and hence
    $\llb \mathrm{sum} \rrb = \llb 0 \rrb$.
    \item If $\mathbf{P}$ is invalid, then at least one coefficient
    $1 - \sum_{i=1}^k \mathbf{P}(i, j)$ is nonzero, making
    $\llb \mathrm{sum} \rrb$ a nonzero polynomial in $\lambda$ of degree
    at most $k-1$.
\end{itemize}

A nonzero polynomial of degree $k-1$ has at most $k-1$ roots in
$\mathbb{K}$. Since $\lambda$ is uniform over $\mathbb{K}$, the
probability that $\lambda$ is a root is
$$ p_\mathrm{fail} \leq \frac{k-1}{|\mathbb{K}|} \leq \frac{k-1}{2^\kappa}. $$
Thus, $\llb \mathrm{sum} \rrb \neq \llb 0 \rrb$, and the protocol aborts
with probability
$$ p = 1 - p_\mathrm{fail} \geq 1 - \frac{k-1}{2^\kappa}. $$
\end{proof}

\begin{corollary}\label{corollary:Security of PermMat}
The protocol $\mathrm{PermMat}$ described in Algorithm~\ref{alg:PermMat}
is secure in the following sense:
\begin{itemize}
    \item If all parties act honestly, then the honest parties share a
    valid permutation matrix $\llb \mathbf{P} \rrb$, with the permutation
    chosen by and known to $P_w$.
    \item If the malicious adversary misbehaves, then the protocol aborts
    with overwhelming probability.
\end{itemize}

Additionally, if $P_w$ is honest, then corrupted parties learn no
information about $\mathbf{P}$.
\end{corollary}
\begin{proof}
This follows immediately from Theorem~\ref{th:Security of OneHot} and
Theorem~\ref{th:Security of Permutation Matrix Check}.
\end{proof}

\begin{corollary}\label{cor:Security of SharePerm}
The protocol $\mathrm{SharePerm}$ described in
Algorithm~\ref{alg:SharePerm} is secure in the following sense:
\begin{itemize}
    \item Any adversarial misbehavior leads to an abort with overwhelming
    probability.
    \item If all parties act honestly, then the honest parties share
    $\llb \pi \rrb$, where $\pi$ is chosen by and known to $P_w$.
\end{itemize}

Additionally, if $P_w$ is honest, then the adversary learns no information
about $\pi$.
\end{corollary}
\begin{proof}
By Corollary~\ref{corollary:Security of PermMat}, if the protocol has not
aborted, then all shared matrices are valid permutation matrices, which
implies that the permutation is correctly shared.
\end{proof}

\begin{corollary}\label{cor:Security of Permute}
The protocol $\mathrm{Permute}(\llb \pi \rrb, \llb \mathbf{v} \rrb)$
described in Algorithm~\ref{alg:Permute} is secure in the following sense:
\begin{itemize}
    \item Any adversarial misbehavior leads to an abort with overwhelming
    probability.
    \item If all parties act honestly, then the honest parties share
    $\llb \pi(\mathbf{v}) \rrb$.
\end{itemize}

This follows directly from the definition of a shared permutation and the
security of the ideal functionalities.
\end{corollary}

This concludes the security proof for our permutation protocol.

\subsection{Security of the Shuffle Protocol}
Finally, we establish the security of the shuffle protocols by combining
Theorem~\ref{th:security of shuffle},
Corollary~\ref{cor:Security of SharePerm}, and
Corollary~\ref{cor:Security of Permute}.

\begin{corollary}
The protocols $\mathrm{Shuffle1}_\mathrm{off}$ and
$\mathrm{Shuffle1}_\mathrm{on}$ described in
Algorithms~\ref{alg:FirstShuffleOff} and~\ref{alg:FirstShuffleOn} are
secure in the following sense:
\begin{itemize}
    \item They abort with overwhelming probability if the adversary
    misbehaves.
    \item Otherwise, they shuffle $\llb \mathbf{v} \rrb$ using a secret
    permutation that is uniformly random over all $m$-permutations.
\end{itemize}
\end{corollary}
\begin{proof}
This follows directly from Corollary~\ref{cor:Security of SharePerm} and
Corollary~\ref{cor:Security of Permute}.
\end{proof}

\begin{corollary}
The protocols $\mathrm{Shuffle2}_\mathrm{off}$ and
$\mathrm{Shuffle2}_\mathrm{on}$ described in
Algorithms~\ref{alg:shuffle off} and~\ref{alg:shuffle on} are secure in
the following sense:
\begin{itemize}
    \item They abort with overwhelming probability if the adversary
    misbehaves.
    \item Otherwise, they shuffle $\llb \mathbf{v} \rrb$ using a secret
    permutation that is uniformly random over all $m$-permutations.
\end{itemize}
\end{corollary}
\begin{proof}
This follows directly from Theorem~\ref{th:security of shuffle},
Corollary~\ref{cor:Security of SharePerm}, and
Corollary~\ref{cor:Security of Permute}.
\end{proof}

\section{Proofs of Complexity Analysis} \label{sec:complexity}

\newcommand{\rand}{\Pi_\mathrm{rand}}
\subsection{Proofs for Section \ref{subsec:complexity of sec4}} \label{subsec:proof of sec4}
\begin{theorem}\label{theorem:Complexity of Permute}
    For the protocol $\mathrm{Permute}(\ll \pi\rr_k, \ll \mathbf{v}\rr)$
    described in Algorithm~\ref{alg:Permute},
    the communication, computation, and round complexities are
    $O(\frac{1}{\log k}nm\log m)$,
    $O(\frac{k}{\log k}nm\log m)$,
    and $O(\frac{\log m}{\log k})$, respectively.
\end{theorem}
\begin{proof}
    The parties compute
    $O(\frac{m\log m}{k\log k})$ matrix-vector products,
    each multiplying a $k\times k$ permutation matrix
    by a data vector of length $k$.
    Such a multiplication consists of $k$ calls to $\Pi_\mathrm{inner}$,
    each of which has communication complexity $O(nk)$.
    Hence, the total communication complexity is
    $O(\frac{1}{\log k}nm\log m)$.
    The computation complexity of a single matrix-vector product is
    $O(nk^2)$, and hence the total computation complexity is
    $O(\frac{k}{\log k}nm\log m)$.

    For the round complexity, note that the parties can parallelize
    all small permutations $\pi'_{i, j}$ with the same $i$.
    Thus, the entire $i$-th iteration is completed within $O(1)$ rounds,
    and the round complexity is simply
    $O(s) = O(\frac{\log m}{\log k})$.
\end{proof}

\begin{lemma}\label{lemma:Complexity of BinCheck}
    For the protocol
    $\mathrm{BinCheck}(\ll b_1, ..., b_d\rr)$
    described in Algorithm~\ref{alg:BinCheck},
    the communication, computation, and round complexities are
    $O(n)$, $O(nd)$, and $O(1)$, respectively.
\end{lemma}
\begin{proof}
    The only communication occurs through one call to each of
    $\Pi_\mathrm{challenge}$, $\Pi_\mathrm{open}$, and
    $\Pi_\mathrm{inner}$, all of which incur $O(n)$ communication and
    $O(1)$ rounds.

    The computation bottleneck is computing
    $$\ll u\rr \gets \sum_{i=1}^d \lambda^{i-1}\ll b_i\rr\ll 1-b_i\rr, $$
    which requires $O(nd)$ computation in total.
\end{proof}

\begin{lemma}\label{lemma:Complexity of Demux}
    For the protocol $\mathrm{Demux}(\ll b_1,...,b_d\rr)$
    described in Algorithm~\ref{alg:Demux},
    the communication, computation, and round complexities are
    $O(n2^d)$, $O(n2^d)$, and $O(\log d)$, respectively.
\end{lemma}
\begin{proof}
    Denote by $M(d)$ the number of calls to the primitive
    $\Pi_\mathrm{mul}$. Then
    $$ M(d) = 2 \times M(d/2) + 2^d. $$
    Since, for all $d\geq 4$,
    $$2^{d+1} \geq 2 \times 2^{d/2+1} + 2^d = 2^{d/2+2} + 2^d,$$
    we conclude that $M(d) = O(2^{d+1}) = O(2^d)$.

    Recall that in the protocol $\mathrm{Demux}$, only calls to
    $\Pi_\mathrm{mul}$ require communication. Hence, the communication
    complexity is $O(n2^d)$. The computation complexity follows similarly.

    Denote by $R(d)$ the number of rounds. Note that
    $$ R(d) = R(d/2) + O(1). $$
    This is because the two recursive calls to $\mathrm{Demux}$ can be
    parallelized, and the remaining operations are parallel calls to
    $\Pi_\mathrm{mul}$. Hence, the round complexity is $O(\log d)$.
\end{proof}

    \begin{lemma}\label{lemma:Complexity of OneHot}
        For the protocol $\mathrm{OneHot}(k, P_w : \mathrm{ind})$
            described in Algorithm \ref{alg:OneHot},
the communication, computation, and round complexities are
$O(nk)$, $O(nk)$, and $O(\log\log k)$, respectively.
    \end{lemma}
\begin{proof}
    The bottleneck is the call to
    $\mathrm{Demux}(\ll b_1, ..., b_d\rr)$,
    where $d=\log k$.
    By Lemma~\ref{lemma:Complexity of Demux}, its communication,
    computation, and round complexities are
    $O(n2^d)$, $O(n2^d)$, and $O(\log d)$, respectively.
    Since $d = \log k$, the lemma follows.
\end{proof}

\begin{lemma}\label{lemma:Complexity of PermCheck}
    For the protocol $\mathrm{PermCheck}(\ll \mathbf{P}\rr)$
    described in Algorithm~\ref{alg:PermCheck},
    the communication, computation, and round complexities are
    $O(n)$, $O(n k^2)$, and $O(1)$, respectively.
\end{lemma}
\begin{proof}
    Communication occurs only in
    $\Pi_\mathrm{challenge}$ and $\Pi_\mathrm{open}$,
    each of which costs $O(n)$ communication and $O(1)$ rounds.

    The most computation-heavy step in the algorithm is that each party
    locally computes the sums, which costs $O(nk^2)$ in total.
\end{proof}

\begin{lemma}\label{lemma:Complexity of PermMat}
    For the protocol $\mathrm{PermMat}(P_w : \pi)$
    described in Algorithm~\ref{alg:PermMat},
    the communication, computation, and round complexities are
    $O(nk^2)$, $O(nk^2)$, and $O(\log\log k)$, respectively.
\end{lemma}
\begin{proof}
    The protocol consists of $k$ parallel calls to
    $\mathrm{OneHot}(k, P_w:\pi(i))$
    plus one call to $\mathrm{PermCheck}(\ll \mathbf{P}\rr)$.
    By Lemma~\ref{lemma:Complexity of OneHot} and
    Lemma~\ref{lemma:Complexity of PermCheck},
    this costs $O(nk^2)$ communication,
    $O(nk^2)$ computation, and $O(\log\log k)$ rounds.
\end{proof}
\renewcommand{\thetheorem}{\arabic{section}.\arabic{theorem}}
\setcounter{section}{4}
\setcounter{theorem}{2}
\begin{theorem}
    For the protocol $\mathrm{SharePerm}(k, P_w: \pi)$
    given in Algorithm~\ref{alg:SharePerm},
    the communication, computation, and round complexities are
    $O(\frac{k}{\log k}nm\log m)$,
    $O(\frac{k}{\log k}nm\log m)$, and
    $O(\log \log k)$, respectively.
\end{theorem}
\begin{proof}
    Recall that the permutation decomposition of \cite{chase2020secret}
    decomposes an $m$-permutation $\pi$ into
    $s=O(\frac{\log m}{\log k})$ $m$-permutations
    $\pi_1, ..., \pi_s$.
    Each permutation $\pi_i$ is then decomposed into
    $m/k$ permutations $\pi_{i, 1}, ..., \pi_{i, m/k}$,
    each of which is transformed into a $k$-permutation
    $\pi_{i, j}^\prime$.

    To share these permutations, $P_w$ needs to share
    $O(\frac{m\log m}{k\log k})$ $k\times k$ permutation matrices in total.
    By Lemma~\ref{lemma:Complexity of PermMat}, the parallel calls to
    $\mathrm{PermMat}$ result in
    $O(\frac{k}{\log k}nm\log m)$ communication and computation,
    and $O(\log\log k)$ rounds.
\end{proof}

\begin{theorem}
    Consider an input collection
    $\ll \mathbf{V} \rr = \ll \mathbf{v}_1, \ldots, \mathbf{v}_m \rr$
    of $m$ vectors, each of length~$l$, to be jointly shuffled by the parties.
    For our first shuffle protocol described in
    Algorithms~\ref{alg:FirstShuffleOff} and~\ref{alg:FirstShuffleOn},
    the offline communication, computation, and round complexities are
    $O(\frac{k}{\log k}n^2m\log m)$,
    $O(\frac{k}{\log k}n^2m\log m)$, and
    $O(\log \log k)$, respectively.
    The online communication, computation, and round complexities are
    $O(\frac{l}{\log k}n^2m\log m)$,
    $O(\frac{kl}{\log k}n^2m\log m)$, and
    $O(\frac{n\log m}{\log k})$, respectively.
\end{theorem}
\begin{proof}
    Note that all parties can share permutation matrices in parallel
    in the offline phase. Hence, the offline round complexity is simply
    $O(\log\log k)$.

    The remaining complexities follow directly from
    Theorem~\ref{theorem:Complexity of Permute} and the complexity theorem
    for $\mathrm{SharePerm}$.
\end{proof}

\setcounter{section}{2}
\subsection{Proofs for Section \ref{subsec:complexity of sec5}}\label{subsec:proof of sec5}
\setcounter{section}{5}
\setcounter{theorem}{1}
\begin{theorem}
    For $\mathrm{Shuffle2}_\mathrm{off}(P_1:\pi_1, ..., P_n:\pi_n)$
    described in Algorithm~\ref{alg:shuffle off},
    the offline communication, computation, and round complexities are
    $O(\frac{k}{\log k}n^2m\log m)$,
    $O(\frac{k}{\log k}n^2m\log m)$, and
    $O(\frac{\log m}{\log k}+\log\log k)$, respectively.
\end{theorem}
\begin{proof}
    There are $n\cdot m$ calls in total to each of
    $\Pi_\mathrm{rand}$, $\Pi_\mathrm{mul}$, and $\Pi_\mathrm{open}$.
    This requires $O(n^2m)$ communication, which is not the bottleneck.

    The bottleneck in all complexity measures comes from the $n$ calls to
    $\mathrm{SharePerm}$ and $\mathrm{Permute}$, each for permuting an
    $m\times 3$ matrix. Since these calls are performed in parallel,
    by combining the offline and online complexities presented in
    Theorem~\ref{theorem:Complexity of Permute} and the complexity theorem
    for $\mathrm{SharePerm}$, we obtain communication, computation, and
    round complexities of
    $O(\frac{k}{\log k}n^2m\log m)$,
    $O(\frac{k}{\log k}n^2m\log m)$, and
    $O(\frac{\log m}{\log k} + \log\log k)$, respectively.
\end{proof}

\begin{theorem}
    For $\mathrm{Shuffle2}_\mathrm{on}(\ll \mathbf{v}\rr)$
    described in Algorithm~\ref{alg:shuffle on},
    the communication, computation, and round complexities are
    $O(nm+n^2)$, $O(n^2m)$, and $O(n)$, respectively.
    Note that when $n=O(m)$, the communication complexity simplifies to
    $O(nm)$.
\end{theorem}
\begin{proof}
    Note that $\mathrm{CorrCheck}$, when called on a vector of length $m$,
    costs $O(n)$ communication, $O(nm)$ computation, and $O(1)$ rounds.
    Besides the calls to $\mathrm{CorrCheck}$, the parties compute
    $\Pi_\mathrm{mul}$ on a vector of length $m$, and each party $P_i$
    sends $\mathbf{y}_i$ of length $O(m)$ to $P_{i+1}$.
    Finally, $P_n$ broadcasts $\mathbf{y}_n$, which requires $O(nm)$
    communication.

    Summing these costs, the communication complexity is $O(n(n+m))$,
    the computation complexity is $O(n^2m)$ due to the $n$ calls to
    $\mathrm{CorrCheck}$, and the round complexity is $O(n)$.
\end{proof}

\setcounter{section}{2}

\section{Discussion}\label{sec:discussion}

\subsection{Sharing Permutation Matrix}\label{sec:Generating Permutation Matrix}
Recall that the parties wish to perform computations over a field
$\mathbb{F}$ with prime subfield $\mathbb{F}_p$. Suppose party $P_w$
needs to share an $m\times m$ permutation matrix. When $p>m$, generating
and checking a permutation matrix can be much easier.

To generate a purported permutation matrix chosen by party $P_w$, the
protocol simply requires $P_w$ to share the matrix. This matrix is a valid
permutation matrix if the following conditions are satisfied:
\begin{enumerate}
    \item Each entry of the matrix is either $0$ or $1$.
    \item The sum of the entries in each row is exactly $1$.
    \item The sum of the entries in each column is exactly $1$.
\end{enumerate}
To see why this holds, note that the first condition ensures that the
matrix is binary. Since $p>m$, the sum in any row or column does not
suffer from modular overflow; that is, the sum directly equals the number
of $1$s in that row or column. Hence, a binary matrix with exactly one
$1$ in each row and each column is, by definition, a permutation matrix.
Notably, the first condition can be verified with a single invocation of
$\inner$, and all three conditions can be batched into a single invocation
of $\open$. Thus, this check incurs almost the same overhead as the
$\mathrm{PermCheck}$ protocol.

However, the above argument is valid only when $p>m$. A counterexample
arises when $m=p+1$: a malicious party could share an all-ones matrix and
pass all the checks, since $p+1 \equiv 1 \pmod{p}$ in such a field.
In the case of $p\leq m$, it is significantly harder to check whether an
arbitrary matrix is a valid permutation matrix, because modular overflow
makes row and column sum checks unreliable.

Hence, in our construction, we impose specific constraints via one-hot
vectors to simplify such checks in the general case, for arbitrary $p$ and
$m$. The asymptotic communication complexities of the two approaches,
namely the direct check for $p>m$ and the one-hot vector-based check for
general $p$, are identical: both are $O(nm^2)$ per permutation matrix for
generation and verification.

\subsection{Online Complexity}\label{subsec:Online Complexity}
There are several reasons why the online complexity of an MPC protocol is
more important than its offline complexity. First, if we build an
anonymous communication system based on an MPC shuffle, such as
Clarion~\cite{eskandarian2021clarion}, the online complexity corresponds
to the latency experienced by users or clients. Similar reasoning applies
to other real-world MPC applications, such as private data analysis and
secure multiparty computation for financial transactions.

Another important reason is that the offline phases of multiple protocol
sessions can be executed in parallel. This allows a significant reduction
in the round complexity of the overall execution. For example, suppose the
parties aim to perform a shuffle sequentially $k$ times, where each
execution requires $R_\mathrm{off}$ rounds in the offline phase and
$R_\mathrm{on}$ rounds in the online phase. With naive parallelization of
data-independent operations, the total execution requires only
$R_\mathrm{off}$ offline rounds for all $k$ sessions and
$k \cdot R_\mathrm{on}$ online rounds, one after another for each session.
This effect is even more pronounced if the shuffle protocol serves as a
subroutine of a higher-level protocol.

A concrete example is the MPC radix sort algorithm of
\citet{hamada2014radixsort}, which relies on MPC shuffle as a core
primitive. To sort data by 128-bit indices, 128 sequential invocations of
the shuffle protocol are required, resulting in at least
$R_\mathrm{off} + 128 \times R_\mathrm{on}$ rounds. Since only the online
phase is inevitably embedded as a subroutine in higher-level protocols,
only the online round complexity is subject to this kind of amplification.
This makes the online round complexity much more important than its
offline counterpart.

\subsection{Definition of Shared Permutation}\label{subsec:Definition of Shared Permutation}
The requirement that
$s = O\left(\frac{\log m}{\log k}\right)$ in
Definition~\ref{definition: shared permutation} is tight, meaning that
the bound cannot be asymptotically improved. To see this, note that there
are $m!$ distinct $m$-permutations. The number of possible $k$-shared
$m$-permutations, as defined earlier, is at most $(k!)^{ms/k}$. This
number must be at least $m!$ to cover all $m$-permutations. Hence,
$$
s \geq \log_{(k!)^{m/k}} (m!) = \frac{k\log(m!)}{m\log (k!)}.
$$
By Stirling's approximation, $\log(m!) = \Theta(m\log m)$ and
$\log(k!) = \Theta(k\log k)$. Substituting these estimates into the
inequality gives
$s = \Omega\left(\frac{\log m}{\log k}\right)$, which matches the upper
bound $s = O\left(\frac{\log m}{\log k}\right)$ and proves tightness.

Note that one may also consider a recursive definition, where
$\llb \mathbf{P}_{i,j} \rrb$ is replaced by a recursive
$\llb \pi_{i,j} \rrb$. The recursion terminates when a given
$\llb \pi \rrb$ is shared exactly as $\llb \mathbf{P}_\pi \rrb$, namely
as a direct permutation matrix share. Nevertheless, the current definition
is sufficient for our purposes; recursion would add unnecessary complexity
without improving performance.

\subsection{Overhead of Malicious and GOD-Enabled Realizations}
\label{app:implementation-scope}
\rev{
This subsection clarifies the scope of the benchmark implementation and
estimates the overhead of full malicious and GOD-enabled realizations.
The experiments in Section~\ref{subsec:Evaluation} instantiate Shamir
secret sharing in the semi-honest security model. They are intended to
measure the core execution cost of our shuffle constructions, rather than
the concrete running time of a complete malicious or GOD-enabled
implementation. The following multiplicative factors should be viewed as
optimistic but reasonable protocol-level estimates based on the additional
vectors, checks, and primitive calls required by the malicious or
GOD-enabled protocols. They are not separate benchmark measurements, nor
should they be interpreted as lower bounds; the actual overhead may be
higher depending on the concrete malicious Shamir backend and
implementation details.
}

\rev{
For \textsf{Shuffle1}, the implementation directly generates and shares
the decomposed $k\times k$ permutation matrices. It does not implement the
malicious well-formedness checks for these matrices, namely
\textsf{OneHot}, \textsf{BinCheck}, and \textsf{PermCheck}. A malicious
implementation would add these checks in the offline phase. Since both
direct matrix sharing and checked matrix generation have $O(nk^2)$
communication per small permutation matrix, the asymptotic bound is
unchanged. Concretely, we estimate the offline cost of \textsf{Shuffle1}
to increase by about $2\times$, mainly due to additional secure
multiplications, openings, and the $O(\log\log k)$-round demultiplexing
procedure. The online phase of \textsf{Shuffle1} is unchanged at the
protocol-structure level, except that the underlying Shamir primitives
must be replaced by actively secure ones. Thus, its online cost is
estimated to increase by about $1.5\times$, depending on the concrete
malicious Shamir backend.
}

\rev{
For \textsf{Shuffle2}, the benchmark implements a semi-honest
shuffle-correlation chain. In the offline phase, it preprocesses masks
$r_i$ and the permuted masks $\pi_i(r_i)$. A malicious implementation
requires additional correlated values based on a hidden random scalar
$\beta$. In particular, it must preprocess $\beta r_i$ and $r'_i$, and
apply each party's permutation to three vectors
$(r_i,\beta r_i,r'_i)$. It also opens two-component correction values
$(z_{i,1},z_{i,2})$ instead of the single correction vector used in the
semi-honest implementation. Therefore, the malicious offline phase of
\textsf{Shuffle2} is estimated to be about $3\times$ more
expensive than the benchmarked semi-honest version.
}

\rev{
The online phase of malicious \textsf{Shuffle2} also carries a larger
constant factor. The semi-honest implementation sends a single vector
$y_i$ along the chain, whereas the malicious protocol sends a two-component
message $y_i=(y_{i,1},y_{i,2})$ and invokes \textsf{CorrCheck} at each
hop. In addition, the first online step computes the $\beta$-masked input
used for the correlation check. Hence, the main online payload is roughly
doubled, and the correlation-checking layer adds secure multiplications,
openings, and rounds. Overall, we estimate the online cost of malicious
\textsf{Shuffle2} to increase by about $2\times$ relative to the
semi-honest benchmark. The dominant online communication nevertheless
remains $O(nml)$, and the number of online rounds remains $O(n)$.
}

\rev{
The GOD-enabled extension adds further overhead beyond malicious security
with abort. It requires GOD-enabled Shamir primitives, dispute-control
state such as $\mathrm{Corr}$ and $\mathrm{Disp}$, and authentication tags
for intermediate shares. Existing honest-majority GOD protocols suggest
that this overhead can be kept small: for example, in the framework of
\citet{goyal2020guaranteed}, the concrete communication per multiplication
gate ranges from $5.5$ to $7.5$ field elements per party. This is close to
the security-with-abort result of \citet{goyal2020malicious}, which
requires $5.5$ field elements per multiplication gate per party and has
essentially the same concrete communication cost as semi-honest protocols
in this setting.
}

\rev{
Thus, enabling GOD at the arithmetic layer is expected to introduce only a
small backend-dependent constant-factor overhead in fault-free executions.
The end-to-end overhead of \textsf{SLIDE} may be larger because our
construction also maintains dispute-control state
\cite{beerliova2006dispute} and authentication tags
\cite{Ben-Sasson2012miniMPC}. If adversarial deviations are detected,
dispute control may trigger relay communication and rollback within the
current segment. Such recovery events can be more expensive, but each
event identifies a new corrupted party or a new disputing pair, so the
number of such events is globally bounded; see
Section~\ref{subsec:Round Complexity under Dispute Control}
}

\rev{
In summary, compared with the benchmarked semi-honest implementation, a
full malicious implementation is expected to increase the concrete cost of
\textsf{Shuffle1} by roughly $2\times$ offline and $1.5\times$ online, and
the concrete cost of \textsf{Shuffle2} by roughly $3\times$ offline and
$2\times$ online. Adding GOD support is expected to introduce only
a small backend-dependent constant-factor overhead in fault-free executions,
consistent with existing honest-majority GOD MPC frameworks
\cite{goyal2020guaranteed}. Overall, these overheads are expected to be
modest constant-factor increases.
}

\subsection{Round Complexity under Dispute Control}
\label{subsec:Round Complexity under Dispute Control}
When achieving guaranteed output delivery (GOD) via dispute control,
rollbacks may introduce additional rounds. If only a single shuffle is
executed in isolation, the worst-case round complexity can reach
$O(n^2)$. However, in typical MPC applications, the shuffle protocol is
invoked repeatedly as a subroutine within a full protocol. In such
protocols, all secure primitives follow the same dispute control mechanism,
and the total number of rollbacks is bounded by $O(n^2)$ in the worst
case. Since parties identified as faulty are excluded from subsequent
computations, these rollbacks do not accumulate across invocations.
Therefore, the additional rounds should be attributed to the overall cost
of achieving GOD in the full protocol, rather than to any individual
subroutine. In particular, the round complexity of a single shuffle
remains $O(n)$.

\end{document}